\def\version{November 25, 2014}

\documentclass[12pt]{article}
\def\macrosPb{} 
\def\macrosHarxiv{} 
\usepackage[psamsfonts]{amsfonts}
\usepackage{amsmath,amssymb,amsthm}
\usepackage[dvips]{graphicx}
\usepackage{appendix}
\usepackage{bbm} 
\usepackage{amsbsy}
\usepackage{enumerate}
\usepackage{cite}

\ifdefined\macrosPa
  \usepackage[textwidth=465pt,textheight=650pt,centering]{geometry} 
\else\ifdefined\macrosPb
  \usepackage[textwidth=500pt,textheight=650pt,centering]{geometry} 
\fi\fi

\ifdefined\macrosS
  \makeatletter
  
  \def\boldsymbol{\pmb}
  \makeatother

  \usepackage{mathptmx}
  \DeclareMathAlphabet{\mathcal}{OMS}{cmsy}{m}{n}
\fi

\def\UseSection{
        \numberwithin{equation}{section}
	\theoremstyle{plain}
        \newtheorem{theorem}    {Theorem}[section]
        \DefineTheorems 
}

\def\DefineTheorems{
	
	\newtheorem{lemma}      [theorem] {Lemma}
	
	\newtheorem{prop}       [theorem] {Proposition}
	
	\newtheorem{cor}        [theorem] {Corollary}

	\theoremstyle{definition}
	\newtheorem{defn}       [theorem] {Definition}
	
	\newtheorem{example}       [theorem] {Example}

	\theoremstyle{definition}

}

\newcommand{\bt}   {\begin{theorem}}
\newcommand{\et}   {\end  {theorem}}
\newcommand{\bl}   {\begin{lemma}}
\newcommand{\el}   {\end  {lemma}}
\newcommand{\bp}   {\begin{prop}}
\newcommand{\ep}   {\end  {prop}}
\newcommand{\bc}   {\begin{cor}}
\newcommand{\ec}   {\end  {cor}}
\newcommand{\bd}   {\begin{defn}}
\newcommand{\ed}   {\end  {defn}}

\newcommand{\ba}   {\begin{array}}
\newcommand{\ea}   {\end  {array}}
\newcommand{\be}   {\begin{enumerate}}
\newcommand{\ee}   {\end  {enumerate}}
\newcommand{\bi}   {\begin{itemize}}
\newcommand{\ei}   {\end  {itemize}}

\def\eq#1\en{\begin{equation}#1\end{equation}}  
\def\eqsplit#1\ensplit{
	\begin{equation}\begin{split}#1\end{split}\end{equation}
	}
\def\eqalign#1\enalign{
	\begin{align}#1\end{align}
	}
\def\eqmul#1\enmul{
	\begin{multline}#1\end{multline}
	}
\newcommand{\eqarrstar} {\begin{eqnarray*}} 
\newcommand{\enarrstar} {\end{eqnarray*}} 
\newcommand{\eqarray}   {\begin{eqnarray}} 
\newcommand{\enarray}   {\end{eqnarray}} 
\newcommand{\nnb}	{\nonumber \\} 

\newcommand{\lbeq}[1]  {\label{e:#1}}
\newcommand{\refeq}[1] {\eqref{e:#1}}    

\makeatletter
\newcommand{\labelcounter}[2]{{%
	\stepcounter{#1}
	\protected@write\@auxout{}%
	{\string\newlabel{#2}{{\csname the#1\endcsname}{\thepage}}}%
	{\ref{#2}}
	}}
\makeatother
\newcommand{\Cbold} {{\mathbb C}}

\newcommand{\Nbold} {{\mathbb N}}

\newcommand{\Rbold} {{\mathbb R}}

\newcommand{\Zbold} {{\mathbb Z}}

\newcommand{\Ecal}   {\mathcal{E}}

\newcommand{\Ical}   {\mathcal{I}}

\newcommand{\Mcal}   {\mathcal{M}} 
\newcommand{\Ncal}   {\mathcal{N}} 
 
\newcommand{\Pcal}   {\mathcal{P}}

\newcommand{\Rcal}   {\mathcal{R}}

\newcommand{\Ucal}   {\mathcal{U}} 
\newcommand{\Vcal}   {\mathcal{V}}

\newcommand{\Rd}    {{ {\Rbold}^d}}
\newcommand{\Zd}    {{ {\Zbold}^d }}

\newcommand{\spose}[1] {{\hbox to 0pt{#1\hss}} }
\newcommand{\ltapprox} {\mathrel{\spose{\lower 3pt\hbox{$\mathchar"218$}}
 \raise 2.0pt\hbox{$\mathchar"13C$}}}
\newcommand{\gtapprox} {\mathrel{\spose{\lower 3pt\hbox{$\mathchar"218$}}
 \raise 2.0pt\hbox{$\mathchar"13E$}}}

\UseSection   
\usepackage[usenames]{color}

\definecolor{at}{rgb}{0.0, 0.5, 0.0} 

\newcommand{\Tay}{{\rm Tay}}

\newcommand{\LTsym}{{\rm loc}}
\newcommand{\LT}{{\rm Loc}  }

\renewcommand{\to} {\rightarrow}
\renewcommand{\qed}{\hfill\rule{2mm}{2mm}\bigskip}

\newcommand{\R}{\Rbold}
\newcommand{\Z}{\Zbold}

\newcommand{\N}{\Nbold}
\newcommand{\C}{\mathbb{C}}

\newcommand{\Lambdabold}{\boldsymbol{\Lambda}}

\newcommand{\psib}{\bar\psi}
\newcommand{\ci}{\underline{i}}

\newcommand{\pair}[1]{\langle #1 \rangle}

\newcommand{\Phipol}{\Pi}

\newcommand{\units}{\Ucal}

\newcommand{\cgam}{\gamma}

\newcommand{\h}{\mathfrak{h}}

\newcommand{\Id}{\mathrm{Id}}

\newcommand{\pp}{a}
\newcommand{\qq}{b}
\newcommand{\sigmaa}{\sigma}
\newcommand{\sigmab}{\bar{\sigma}}

\newcommand{\sgn}{\mathrm{sgn}}

\ifdefined\macrosH
  \usepackage{xr-hyper}
  \usepackage{hyperref}
  \hypersetup{hypertexnames=false}
  \hypersetup{colorlinks,citecolor=blue,linkcolor=blue}  

\else\ifdefined\macrosHarxiv
  \usepackage{xr-hyper}
  \usepackage{hyperref}
  \hypersetup{hypertexnames=false}

\else
  \newcommand{\texorpdfstring}[2]{#1}
  \usepackage{xr}
\fi\fi

\title  {
       A renormalisation group method.
       \\
       II.
       Approximation by local polynomials
        }

\author{
David C. Brydges\thanks{Department of Mathematics,
University of British Columbia,
Vancouver, BC, Canada V6T 1Z2.
E-mail: {\tt db5d@math.ubc.ca}, {\tt slade@math.ubc.ca}.}\;
 and Gordon Slade$^*$}

\date\version

\begin{document}

\maketitle

\begin{abstract}
This paper is the second in a series devoted to the development of a
rigorous renormalisation group method for lattice field theories
involving boson fields, fermion fields, or both.  The method is set
within a normed algebra $\Ncal$ of functionals of the fields.  In this
paper, we develop a general method---localisation---to approximate an
element of $\Ncal$ by a local polynomial in the fields.  From the
point of view of the renormalisation group, the construction of the
local polynomial corresponding to $F \in \Ncal$ amounts to the
extraction of the relevant and marginal parts of $F$.  We prove
estimates relating $F$ and its corresponding local polynomial, in
terms of the $T_\phi$ semi-norm introduced in part~I of the series.
\end{abstract}


\section{Introduction and main results}
\label{sec:intro}

This paper is the second in a series devoted to the development of a
rigorous renormalisation group method.  In \cite{BS-rg-norm}, we
defined a normed algebra $\Ncal$ of functionals of the fields.  The
fields can be bosonic, or fermionic, or both, and in
most of this paper there
is no distinction between these possibilities.  The algebra $\Ncal$ is
equipped with
the $T_\phi$ semi-norm, which is defined in terms of a normed space $\Phi$
of test functions.  In the renormalisation group method, a sequence of
test function spaces $\Phi_j$ is chosen, with corresponding normed
algebras $\Ncal_j$, and there is a dynamical system whose trajectories
evolve through these normed algebras in the sequence $\Ncal_0 \to
\Ncal_1 \to \Ncal_2 \to \cdots$.  The dimension of the dynamical
system is unbounded, but a finite number of local polynomials in the
fields represent the relevant (expanding) and marginal (neutral)
directions for the dynamical system.  These local polynomials play a
central role in the renormalisation group approach.

In this paper, we develop a general method for the extraction from an
element $F \in \Ncal$ of a local polynomial $\LT_XF$, localised on a
spatial region $X$, that captures the relevant and marginal parts of
$F$.  We also prove norm estimates which show that the norm of
$\LT_XF$ is not much larger than the norm of $F$, while the norm of
$F-\LT_XF$ is substantially smaller than the norm of $F$.  The latter
fact, which is crucial, indicates that $\LT_XF$ has encompassed the
important part of $F$, leaving the irrelevant remainder $F-\LT_XF$.
The method used in our construction of $\LT_XF$ bears some relation to
ideas in \cite{BR92}.

This paper is organised as follows.  Section~\ref{sec:intro} contains
the principal definitions and statements of results,
as well as some of the simpler proofs.  More substantial
proofs are deferred to Section~\ref{sec:LTsym}.
Section~\ref{sec:Taylor} contains estimates for lattice Taylor
expansions; these play an essential role in the proofs of
Propositions~\ref{prop:Locbd}--\ref{prop:LTKbound}, which provide the
norm estimates on $\LT_XF$ and $F-\LT_XF$.

\subsection{Fields and test functions}

We recall some concepts and notation from \cite{BS-rg-norm}.

Let $\Lambda = \Zd/(mR\Z)$ denote the $d$-dimensional discrete torus
of (large) \emph{period} $mR$ for integers $R \ge 2$ and $m \ge 1$.
In \cite{BS-rg-norm}, we introduced an index set $\Lambdabold =
\Lambdabold_b \sqcup \Lambdabold_f$.  The set $\Lambdabold_b$ is
itself a disjoint union of sets $\Lambdabold_b^{(i)}$ ($i=1,\ldots,
s_b$) corresponding to different species of boson fields.  Each
$\Lambdabold_b^{(i)}$ is either a finite disjoint union of copies of
$\Lambda$, with each copy representing a distinct field component for
that species, or is $\Lambda \sqcup \bar\Lambda$ when a complex field
species is intended.  The set $\Lambdabold_f$ has the same structure,
with possibly a different number $s_f$ of fermion field species.

An element of $\R^{\Lambdabold_b}$ is called a \emph{boson field}, and
can be written as $\phi = (\phi_{x})_{x \in \Lambdabold_{b}}$.  Let
$\Rcal=\Rcal(\Lambdabold_b)$ denote the ring of functions from
$\R^{\Lambdabold_b}$ to $\C$ having at least $p_\Ncal$ continuous
derivatives, where $p_\Ncal$ is fixed.  The \emph{fermion field} $\psi
= (\psi_{y})_{y \in \Lambdabold_{f}}$ is a set of anticommuting
generators for an algebra $\Ncal=\Ncal (\Lambdabold)$ over the ring
$\Rcal$.  By definition (see \cite{BS-rg-norm}),
$\Ncal$ consists of elements $F$ of the form
\begin{equation}
    \label{e:K}
    F
=
    \sum_{y \in \vec\Lambdabold_f^*} \frac{1}{y!} F_y \psi^y
,
\end{equation}
where each coefficient $F_{y}$ is an element of $\Rcal$.  We will use
test functions $g : \vec\Lambdabold^* \to \C$ as defined in
\cite{BS-rg-norm}.  Also, given a boson field $\phi$, we will use the
bilinear pairing between elements of $\Ncal$ and test functions
defined in \cite{BS-rg-norm} and written as
\begin{equation}
\label{e:pairdef}
    \pair{F,g}_\phi = \sum_{z \in \vec\Lambdabold^*}
    \frac{1}{z!} F_z(\phi)g_z.
\end{equation}

For our present purposes, we distinguish between the
boson and fermion fields only through the dependence of
the pairing on the boson field $\phi$.
When the distinction is unimportant,
we use $\varphi$ to denote both kinds of fields,
and identify $\vec\Lambdabold$ with $\Lambda \times \{1,2,\dots
,p_{\Lambdabold} \}$, where $p_{\Lambdabold}$ is the number
of copies of $\Lambda$ comprising $\vec\Lambdabold$.
This $p_{\Lambdabold}$ is
given by the sum, over all species, of the number of components
within a species.
Thus we can write the fields all evaluated at $x \in \Lambda$
as the sequence $\varphi (x) =
(\varphi_{1}(x),\dots,\varphi_{p_{\Lambdabold}}(x))$.

\subsection{Local monomials and local polynomials}

Let $e_{1},\dots ,e_{d}$ denote the standard unit vectors in $\Zd$, so that
\begin{equation}
\label{e:units}
    \units = \{\pm e_{1},\dots ,\pm e_{d}\}
\end{equation}
is the set of all $2d$ unit vectors.
For $e \in \units$ and $f:\Lambda\to\C$,
the difference operator is
given by
\begin{equation}
    \nabla^{e} f(x)=f(x+e) - f(x).
\end{equation}
When $e$ is one of
the standard unit vectors $\{e_{1},\dots ,e_{d} \}$, we refer to
$\nabla^{e}$ as a \emph{forward derivative}. When $e$ is the negative
of a standard unit vector we refer to $\nabla^{e}$ as a \emph{backward
derivative}, although it is the negative of a conventional backward
derivative.  We allow $2d$ directions in $\units$, rather than only
$d$, so as not to break lattice symmetries by favouring forward
derivatives over backward derivatives. This introduces redundancy
expressed by the identity
\begin{equation}
    \label{e:Vcal-relation}
    \nabla^{e}  + \nabla^{-e}
    =
    - \nabla^{-e}\nabla^{e}
,
\end{equation}
which is straightforward to verify by evaluating both sides on a
function $f$.  For $\alpha \in \N_0^\units$ with components $\alpha(e)
\in \N_0$, we write
\begin{equation}
    \label{e:nabla-alpha}
    \nabla^\alpha  =
    \prod_{e \in \units} \nabla^{\alpha(e)} ,
\quad\quad
    \nabla^{0} = \Id
,
\end{equation}
where the product is independent of the order of its factors.

A
\emph{local monomial} $M$ is a finite product of fields and their
derivatives, all to be evaluated at the same point in $\Lambda$ (whose
value we suppress). To be more precise, for $m= (m_{1},\dots ,m_{p
(m)})$ a finite sequence whose components $m_{k} = (i_{k},\alpha_{k})$
are elements of $\{1,\dotsc ,p_{\Lambdabold} \}\times \N_0^\units$, we
define
\begin{equation}
    \label{e:Mm}
    M_{m}
=
    \prod_{k=1}^{p (m)}
    \nabla^{\alpha_{k}}\varphi_{i_{k}}
    =
    \big(\nabla^{\alpha_{1}}\varphi_{i_{1}}\big)
    \cdots
    \big(\nabla^{\alpha_{p(m)}}\varphi_{i_{p(m)}}\big)
.
\end{equation}
The product in $M_m$ is taken in the same order as the components
$i_k$ in $m$.  For example, if the sequence $m$ is given by
$m=((1,\alpha_1),(1,\alpha_1),(1,\alpha_2),(1,\alpha_2),
(1,\alpha_2),(2,\alpha_3))$ with $\alpha_1 < \alpha_2$, then
\begin{equation}
    \label{e:Mmex}
    M_m = (\nabla^{\alpha_1}\varphi_1)^2
    (\nabla^{\alpha_2}\varphi_1)^3 \nabla^{\alpha_3}\varphi_2.
\end{equation}
It is convenient to denote the number of times $m$ contains a given
pair $(i,\alpha)$ as $n_{(i,\alpha)}=n_{(i,\alpha)}(m)$; in
\eqref{e:Mmex} we have $n_{(1,\alpha_1)}=2$, $n_{(1,\alpha_2)}=3$,
$n_{(2,\alpha_3)}=1$, and all other $n_{(i,\alpha)}$ are zero.  For a
fermionic species $i$, $M_{m}=0$ when $n_{(i,\alpha)} > 1$.
Permutations of the order of the components of $m$ give plus or minus
the same monomial. We will now define a subset $\mathfrak{m}$ of
sequences such that every non-zero monomial \eqref{e:Mm} is
represented by exactly one $m \in \mathfrak{m}$. First we fix an order
$\le$ on the elements of $\N_0^\units$.  Let $\mathfrak{m}$ be the set
whose elements are finite sequences as defined above and such that:
(i)~$i_{1} \le \cdots \le i_{p (m)}$; (ii)~for $i$ a fermionic species
$n_{(i,\alpha)} = 0,1$; (iii)~for $k<k'$ with $i_{k}=i_{k'},
\alpha_{k}\le \alpha_{k'}$.
Conditions (i) and (iii) together amount to
imposing lexicographic order on the components of a sequence $m$.

The \emph{degree} of a local monomial $M_{m}$ is the length $p=p (m)$
of the sequence $m\in \mathfrak{m}$.  For $m$ equal to the empty
sequence $\varnothing$ of length $0$, we set $M_\varnothing = 1$,
and we include $m=\varnothing$ in $\mathfrak{m}$.
In addition, we specify a map which associates to each field species a
value in $(0,+\infty]$ called the \emph{scaling dimension} (also
known as \emph{engineering dimension}), which we abbreviate as the
\emph{dimension} of the field species.  Following tradition, for $i =
1,\ldots, p_{\Lambdabold}$, we denote the dimension of the species of
the field $\varphi_i$ by $[\varphi_i]$.  This dimension does
\emph{not} depend on the value of the field, only on its species.
Then we define the \emph{dimension} of $M_{m}$ by
\begin{equation}
\label{e:dimdef}
    [M_{m}]
=
    \sum_{k=1}^{p (m)} \big([\varphi_{i_{k}}] + |\alpha_{k}|_1 \big)
,
\end{equation}
with the degenerate case $[M_\varnothing]=[1]=0$.

Let $\mathfrak{m}_+$ denote the subset of $\mathfrak{m}$ for which
only forward derivatives occur.  Given $d_{+} \ge 0$, let $\Mcal_+$
denote the set of monomials $M_{m}$ with $m\in \mathfrak{m}_+$, such
that
\begin{equation}
    \label{e:Mcal}
    [M_{m}]
\le
    d_{+}.
\end{equation}

\begin{example}
\label{ex:mon}
Consider the case of a single real-valued boson field $\varphi$ of
dimension $[\varphi]=\frac{d-2}{2}$, with no fermion field.
The space $\Ncal_j$ is reached
after $j$ renormalisation group steps have been completed.  Each
renormalisation group step integrates out a fluctuation field, with
the remaining field increasingly smoother and smaller in magnitude.  A
basic principle is that there is an $L>0$ such that
$\varphi_x$ will typically have magnitude approximately
$L^{-j[\varphi]}$, and that moreover $\varphi$ is roughly constant
over distances of order $L^{j}$.  A block $B$ in $\Zd$, of side
$L^{j}$, contains $L^{dj}$ points, so the above assumptions lead to
the rough correspondence
\begin{equation}
    \label{e:relevant-monomials}
    \sum_{x \in B} |\varphi_{x}|^p \approx L^{(d-p[\varphi])j}
    .
\end{equation}
In the case of $d=4$, for which $[\varphi]=1$, this scales down when
$p>4$ and $\varphi^{p}$ is said to be \emph{irrelevant}.  The power
$p=4$ neither decays nor grows, and $\varphi^4$ is called \emph{marginal}.
Powers
$p<4$ grow with the scale, and $\varphi^{p}$ is said to be
\emph{relevant}.
The assumption that $\varphi$ is roughly constant over distances of
order $L^{j}$ translates into an assumption that each spatial derivative
of $\varphi$ produces a factor $L^{-j}$,
so that, e.g.,
$\sum_{x \in B} |\nabla^\alpha\varphi_{x}|^p
\approx L^{(d-p[\varphi]-p|\alpha|_1)j}$.
Thus, in dimension $d=4$ with $d_+=4$,
$\Mcal_+$ consists of the relevant monomials
\begin{equation}
\label{e:locmonrel}
    1,\;\; \varphi,\;\; \varphi^2,\;\; \varphi^3,\;\;
    \nabla_i\varphi,\;\; \nabla_j\nabla_i\varphi,\;\;\varphi \nabla_i \varphi,
\end{equation}
together with the marginal monomials
\begin{equation}
\label{e:locmonmar}
    \varphi^4,\;\;
    \nabla_k\nabla_j\nabla_i \varphi, \;\;
    \varphi \nabla_j\nabla_i \varphi, \;\;
    \varphi^2\nabla_i \varphi,
\end{equation}
with each $\nabla_l$ represents forward differentiation in the
direction $e_l\in \{+e_1,\ldots,+e_d\}$.
\qed
\end{example}

Let $\Pcal$ be the vector space over $\Cbold$ freely generated by all
the monomials $(M_{m})_{ m \in \mathfrak{m}}$ of finite dimension.
A polynomial $P \in \Pcal$ has a unique representation
\begin{equation}
\label{e:P-rep}
    P
=
    \sum_{m \in \mathfrak{m}} a_{m} M_{m}
,
\end{equation}
where all but finitely many coefficients $a_{m} \in \Cbold$ are zero.
Similarly, we define $\Pcal_+$ to be the vector subspace of $\Pcal$
freely generated by the monomials $(M_{m})_{ m \in \mathfrak{m}_+}$ of
finite dimension.
Given $x\in\Lambda$, a polynomial $P\in\Pcal$ is mapped to an element
$P_{x} \in \Ncal$ by evaluating the fields in $P$ at $x$. More
generally, for any $X \subset \Lambda$ and $P \in \Pcal$, we define an
element of $\Ncal$ by
\begin{equation}
    P(X) = \sum_{x \in X} P_x.
\end{equation}

For a real number $t$ we define $\Pcal_{t}$ to be
the subspace of $\Pcal$ spanned by the monomials with $[M_{m}]\ge t$.
Let
\begin{equation}
    \label{e:nufrak+}
    \mathfrak{v}_+ = \{ m \in \mathfrak{m}_+ : [M_m] \le d_+\}
     = \{ m \in \mathfrak{m}_+ : M_m \in \Mcal_+ \},
\end{equation}
and let $\Vcal_+$ denote the vector subspace of $\Pcal_+$ generated by
the monomials in $\Mcal_+$. By definition, the set
$\mathfrak{v}_+$ is finite.  The use of only forward derivatives to
define $\Vcal_+$ breaks the Euclidean symmetry of $\Lambda$.  We wish
to replace $\Vcal_+$ by a symmetric family of polynomials, and this
leads us to consider symmetry in more detail.

Let $\Sigma$ be the group of permutations of $\units$.  Let
$\Sigma_{\text{axes}}$ be the abelian subgroup of $\Sigma$ whose
elements fix $\{e_{i},-e_{i} \}$ for each $i=1,\dots ,d$. In other
words, elements of $\Sigma_{\text{axes}}$ act on $\units$ by possibly
reversing the signs of the unit vectors. Let $\Sigma_{+}$ be the
subgroup of permutations that permute $\{e_1,\ldots, e_d\}$ onto
itself and $\{-e_1,\ldots,-e_d\}$ onto itself.  Then (i)
$\Sigma_{\text{axes}}$ is a normal subgroup of $\Sigma$, (ii) every
element of $\Sigma$ is the product of an element of
$\Sigma_{\text{axes}}$ with an element of $\Sigma_{+}$, and (iii) the
intersection of the two subgroups is the identity. Therefore, by
definition, $\Sigma$ is the semidirect product $\Sigma =
\Sigma_{\text{axes}}\rtimes \Sigma_{+}$.

An element $\Theta\in\Sigma$ acts on elements of $\N_0^\units$
via its action on
components, as $(\Theta\alpha)(e) = \alpha(\Theta(e))$.  The action
of $\Theta$ on derivatives is then given by $\Theta\nabla^\alpha =
\nabla^{\Theta\alpha}$.  This allows us to define an action of the
group $\Sigma$ on $\Pcal$ by linear transformations,
determined by the action
\begin{equation}
    M_{m}
\mapsto
    \Theta M_{m}
=
    \prod_{k=1}^{p (m)}
    \nabla^{\Theta\alpha_{k}}\varphi_{i_{k}}
=
    M_{\Theta m}
\end{equation}
on the monomials, where $\Theta m \in \mathfrak{m}$ is defined by the
action of $\Theta$ on the components $\alpha_k$ of $m$.
We say that $P \in\Pcal$ is
$\Sigma_{\text{axes}}$-\emph{covariant} if there is a homomorphism
$\lambda(\cdot, P):\Sigma_{\text{axes}}\to \{-1,1\}$ such that
\begin{equation}
\label{e:covariant}
    \Theta P
=
    \lambda (\Theta,P)P,
    \quad\quad
    \Theta \in \Sigma_{\text{axes}}
.
\end{equation}
As the notation indicates, the homomorphism can depend on $P$.

The polynomials in $\Vcal_+$ contain only forward derivatives and
hence do not form an invariant subspace of $\Pcal$ under the action of
$\Sigma$.  We wish to replace $\Vcal_+$ by a suitable
$\Sigma$-invariant subspace of $\Pcal$, which we will call $\Vcal$.
As a first step in this process, we define a map that associates to a
monomial $M \in \Mcal_+$ a polynomial $P =P(M) \in \Pcal$, by
\begin{equation}
\label{e:Pm-def}
    P (M)
=
    |\Sigma_{\text{axes}}|^{-1}
    \sum_{\Theta \in \Sigma_{\text{axes}}}
    \lambda (\Theta,M) \Theta M
\end{equation}
where $\lambda (\Theta,M)=-1$ if the number of derivatives in $M$ that
are reversed by $\Theta$ is odd and otherwise $\lambda (\Theta,M)=1$.
This is a homomorphism: for $\Theta,\Theta' \in \Sigma_{\text{axes}}$,
$\lambda (\Theta\Theta',M)=\lambda (\Theta,M)\lambda (\Theta',M)$.
Note that $P(M)$ consists of a linear combination of monomials whose
degrees and dimensions are all equal to those of $M$.  We claim that
for any $M \in \Mcal_+$, the polynomial $P=P(M)$ of
\eqref{e:Pm-def} obeys:
$P (M)$ is $\Sigma_{\text{axes}}$-covariant; $M-P (M) \in \Pcal_{t}$
for some $t>[M]$ up to terms that vanish under the redundancy relation
\eqref{e:Vcal-relation}; and $P (\Theta M)=\Theta P (M)$ for $\Theta
\in \Sigma_{+}$.  The proof of this fact is deferred to
Section~\ref{sec:PM}.

To enable the use of the redundancy relation
\eqref{e:Vcal-relation}, let $\Rcal_{1}$ be the vector subspace of
$\Pcal$ generated by the relation \eqref{e:Vcal-relation}; this is
defined more precisely as follows.  First, $0 \in \Rcal_1$.  Given
nonzero $P \in \Pcal$, we recursively replace any occurrence of
$\nabla^e\nabla^{-e}$ in any monomial in $P$ by the equivalent
expression $-(\nabla^e + \nabla^{-e})$.  This procedure produces
monomials of lower dimension so eventually terminates.  If the
resulting polynomial is the zero polynomial, then $P \in \Rcal_1$, and
otherwise $P \not\in \Rcal_1$. The  claim in the
previous paragraph shows the existence of
the polynomial $\hat P$ of the next definition.

\begin{defn}
\label{def:Vcal} To each monomial $M\in \Mcal_+$ we choose a
polynomial $\hat{P} (M)\in \Pcal$, which is a linear combination of
monomials of the same degree and dimension as $M$, such that
\begin{align}
\label{e:cov-conditions}
(i)
\quad&
\text{$\hat{P}(M)$ is $\Sigma_{\text{axes}}$-covariant},
\nnb
(ii)\quad&
    M-\hat{P} (M)
\in
    \Pcal_{t}   + \Rcal_{1}  \text{ for some $t>[M]$},
\nnb
(iii)\quad&
    \text{$\Theta  \hat{P} (M)
=
    \hat{P} (\Theta M)$ for
    $\Theta  \in \Sigma_{+}$
    }
.
\nonumber
\end{align}
Let $\Vcal$ be the vector subspace of $\Pcal$ spanned by the
polynomials $\{\hat{P} (M): M \in \Mcal_+ \}$.  We also define
$\Vcal(X) = \{P(X) : P\in\Vcal\}$, which is a subset of
$\Ncal$.
\end{defn}

Note that $\Vcal$ depends on our choice of $\hat{P}(M)$
for each $M \in \Mcal_+$,
but is spanned by monomials of dimension at most $d_+$.
The restriction of $\Theta$
to $\Sigma_+$ in
item \emph{(iii)} ensures that $\Theta M \in \Mcal_+$ when $M \in
\Mcal_+$, so that $\hat P(\Theta M)$ makes sense.

\begin{example}
\label{ex:Pm-nonunique}
In practice, we may prefer to choose $\hat P$ satisfying
the conditions of Definition~\ref{def:Vcal} using a formula other
than \eqref{e:Pm-def}.  For example,
for $e \in \units$ let $M_{e} = \varphi \nabla^{e}\nabla^{e} \varphi$.
The formula \eqref{e:Pm-def} gives
\begin{equation}
    P (M_{e})
=
    (1/2)\left(
    \varphi \nabla^{e}\nabla^{e} \varphi +
    \varphi \nabla^{-e}\nabla^{-e} \varphi
    \right)
,
\end{equation}
but via \eqref{e:Vcal-relation} the simpler choice $\hat{P} (M_{e}) =
- \varphi \nabla^{-e}\nabla^{e} \varphi$ also satisfies the conditions
of Definition~\ref{def:Vcal}.
\end{example}

\begin{prop}
\label{prop:EVcal}
The subspace $\Vcal$
is a $\Sigma$-invariant subspace of $\Pcal$.
\end{prop}

\begin{proof}
By Definition~\ref{def:Vcal}\emph{(iii)}, the set $\{\hat{P} (M):M \in
\Mcal_+\}$ is mapped to itself by $\Sigma_{+}$. Since $\hat{P} (M)$ is
$\Sigma_{\text{axes}}$-covariant, $\Vcal$ is invariant under
$\Sigma_{+}$ and $\Sigma_{\text{axes}}$.  Thus, since $\Sigma =
\Sigma_{\text{axes}}\rtimes \Sigma_{+}$, $\Vcal$ is invariant under
$\Sigma$.
\end{proof}

\subsection{The operator \texorpdfstring{$\LTsym$}{loc}}
\label{sec:oploc}

We would like to define polynomial functions on subsets of the torus,
and for this we need to restrict to subsets which do not ``wrap
around'' the torus.  The restricted subsets we use are called
\emph{coordinate patches} and are defined as follows.  Fix a
non-negative integer $p_\Phi \ge 0$ and let $\bar p_\Phi =
\max\{1,p_\Phi\}$.  For a nonempty subset $X\subset\Lambda$, let
$X^{(1)} \supset X$ be the set of all points within $L^{\infty}$
distance $\bar p_\Phi$ of $X$.  This definition is such that the
values of derivatives $\nabla^\alpha g_z$ of a test function $g$ can
be computed when all components of $z$ lie in $X$, for all $\alpha$
with $|\alpha|_\infty \le p_\Phi$, knowing only the values of $g_z$
when all components of $z$ lie in $X^{(1)}$.  For a nonempty subset
$\Lambda '\subset\Lambda$, a map $z = (x_{1}, \dots, x_{d})$ from
$\Lambda '^{(1)}$ to $\Zd$ is said to be a \emph{coordinate} on
$\Lambda '$ if: (i) $z$ is injective and maps nearest-neighbour points
in $\Lambda '^{(1)}$ to nearest-neighbour points in $\Zd$, (ii)
nearest-neighbour points in the image $z (\Lambda '^{(1)})$ of
$\Lambda '^{(1)}$ are mapped by $z^{-1}$ to nearest-neighbour points
in $\Lambda '^{(1)}$.  We say that a nonempty subset $\Lambda '$ of
$\Lambda$ is a \emph{coordinate patch} if there is a coordinate $z$ on
$\Lambda '$ such that $z (\Lambda ')$ is a rectangle $\{x \in
\Zd:|x_{i}| \le r_{i}, i= 1,\dots ,d\}$ for some nonnegative integers
$r_{1},\dots ,r_{d}$.

By ``cutting open'' the torus $\Lambda$, all rectangles with $\max_i 2
(r_{i}+\bar p_\Phi)$ strictly smaller than the period of $\Lambda$ are
clearly coordinate patches.  By definition, the intersection of two
coordinate patches is also a coordinate patch, unless it is empty. If
$z$ and $\tilde{z}$ are both coordinates for a coordinate patch then
there is a Euclidean automorphism $E$ of $\Zd$ that fixes the origin
and is such that $\tilde{z} = Ez$.  This is proved by noticing that
the composition $Z = \tilde{z}\circ z^{-1}$ is well defined on $\{x
\in \Zd:\|x \|_{\infty} \le 1\}$, and therefore $Z$ is a permutation
of the set $\Ucal$ of unit vectors.  The orthogonal transformation $E$
that acts by the same permutation on $\Ucal$ is then an automorphism
of $\Zd$ with the required properties.

Given a coordinate patch $\Lambda'$ with coordinate $z$, and given
$\alpha = (\alpha_{1},\dots ,\alpha_{d})$ in $\Nbold^{d}$, we define
the monomial $z^{\alpha} = x_{1}^{\alpha_{1}} \dots
x_{d}^{\alpha_{d}}$, which is a function defined on $\Lambda'^{(1)}$.
We will define a class of test functions $\Phipol = \Phipol [\Lambda
']$ which are polynomials in each argument by specifying the monomials
which span $\Phipol$. To a local monomial $M_{m} \in \Mcal_{+}$ in
\emph{fields}, as in \eqref{e:Mm}, we associate a monomial $p_{m}\in
\Phipol$ by replacing $\nabla^{\alpha_{k}}\varphi_{i_{k}}$ by
$z_{k}^{\alpha_{k}}$. Thus
\begin{equation}
    \label{e:pmdef}
    p_{m} (z)
=
    \prod_{k=1}^{p (m)} z_{k}^{\alpha_{k}}
,
\end{equation}
which is a function defined on $\Lambda_{i_1}'^{(1)} \times \cdots
\times \Lambda_{i_{p(m)}}'^{(1)}$.  For the degenerate monomial
$m=\varnothing$, we set $p_\varnothing =1$.  We implicitly extend
$p_{m}$ by zero so that it becomes a test function defined on
$\vec\Lambdabold^*$.  For example, we associate the monomial $
z_{1}^{\alpha_1} z_{2}^{\alpha_1} z_{3}^{\alpha_2} z_{4}^{\alpha_2}
z_{5}^{\alpha_2} z_{6}^{\alpha_3}$ to the field monomial
\eqref{e:Mmex}.  However, we will also need the monomial
$z_{1}^{\alpha_2} z_{2}^{\alpha_2} z_{3}^{\alpha_2}
z_{4}^{\alpha_3}z_{5}^{\alpha_1} z_{6}^{\alpha_1}$ which cannot be
obtained from $m \in \mathfrak{m}_{+}$ because the condition (iii)
below \eqref{e:Mmex} now requires $\alpha_{2} \le \alpha_{3} \le
\alpha_{1}$, whereas we chose $\alpha_1<\alpha_2$ in \refeq{Mmex}.
Therefore we define $\bar{\mathfrak{m}}_{+}$ and
$\bar{\mathfrak{v}}_{+}$ by dropping the order condition (iii) in
$\mathfrak{m}_+$ and $\mathfrak{v}_+$.  The space $\Phipol
=\Phipol[\Lambda ']$ is the span of $\{p_{m}: m \in
\bar{\mathfrak{v}}_+\}$.  Euclidean automorphisms of $\Zd$ that fix
the origin act on $\Phipol$ and map it to itself, and therefore
$\Phipol [\Lambda ']$ does not depend on the choice of coordinate on
$\Lambda '$.

An equivalent alternate classification of the monomials in
$\Phipol[\Lambda']$ is as follows.  We define the \emph{dimension} of
a monomial \refeq{pmdef} to be its polynomial degree plus
$\sum_{k=1}^{p} [\varphi_{i_k}]$, i.e.,
$\sum_{k=1}^{p}([\varphi_{i_k}]+|\alpha_k|_1)$, consistent with
\eqref{e:dimdef}.  Then we can define $\Phipol[\Lambda']$ to be the
vector space spanned by the monomials (truncated outside $\Lambda'$ as
above) whose dimension is at most $d_+$.

In the following, we will also need the subspace $S\Phipol$ of
$\Phipol$. This is the image of $\Phipol$ under the symmetry operator
$S$ defined in \cite[Definition~\ref{norm-def:S}]{BS-rg-norm}.

Recall the definition from \cite{BS-rg-norm} that, given $X \subset
\Lambda$, $\Ncal(X)$ consists of those $F \in \Ncal$ such that
$F_z(\phi)=0$ for all $\phi$ whenever any component of $z$ lies
outside of $X$. For nonempty $X \subset \Lambda$, we say $F \in
\Ncal_{X}$ if there exists a coordinate patch $\Lambda '$ such that $F
\in \Ncal (\Lambda')$ and $X \subset \Lambda'$.  The condition $F \in
\Ncal_X$ guarantees that neither $X$ nor $F$ ``wrap around" the torus.

\begin{prop}
\label{prop:LTsymexists} For nonempty $X \subset \Lambda$ and $F \in
\Ncal_X$, there is a unique $V \in \Vcal$, depending on $F$ and $X$,
such that
\begin{equation}
\label{e:LT2pair}
    \pair{F,g}_{0} = \pair{V (X),g}_{0}
    \quad\quad
    \text{for all $g \in \Phipol$}
.
\end{equation}
The polynomial $V$ does not depend on the choice of coordinate $z$ or
coordinate patch $\Lambda'$ implicit in the requirement $F \in
\Ncal_X$, as long as $X \subset \Lambda'$ and $F \in \Ncal
(\Lambda')$.  Moreover, $\Vcal(X)$ and $S\Phipol$ are dual vector
spaces under the pairing $\pair{\cdot,\cdot}_0$.
\end{prop}

The proof of Proposition~\ref{prop:LTsymexists} is deferred to
Section~\ref{sec:LTsym-eu}.  It allows us to define our basic object
of study in this paper, the map $\LTsym_X$.

\begin{defn}
\label{def:LTsym}
For nonempty $X \subset \Lambda$ we define $\LTsym_{X}:
\Ncal_{X} \to \Vcal(X)$ by $\LTsym_X F = V (X)$, where $V$ is the
unique element of $\Vcal$ such that \eqref{e:LT2pair} holds.
For $X = \varnothing$, we define $\LTsym_{\varnothing}=0$.
\end{defn}

By definition, the map $\LT_X : \Ncal_{X} \to \Vcal(X)$ is a linear map.

\subsection{Properties of \texorpdfstring{$\LTsym$}{loc}}
\label{sec:LocDefinition}

By definition, for nonempty $X \subset \Lambda$ and $F \in
\Ncal_{X}$,
\begin{equation}
\label{e:LT2}
    \pair{F,g}_{0} = \pair{\LTsym_XF,g}_{0}
    \quad\quad
    \text{for all $g \in \Phipol$}
.
\end{equation}
Also, if $F=V(X) \in \Vcal(X)$ then trivially $\pair{F,g}_0
=\pair{V(X),g}_0$ and hence the uniqueness in
Definition~\ref{def:LTsym} implies that $\LTsym_XF=V(X)=F$.  Thus
$\LTsym_X$ acts as the identity on $\Vcal(X)$.  The following
proposition shows that $\LTsym$ behaves well under composition.

\begin{prop}
\label{prop:LT2}
For $X,X' \subset \Lambda$ and $F \in
\Ncal_{X\cup X'}$, excluding the case $X'=\varnothing\not= X$,
\begin{align}
    &
    \LTsym_{X}\circ \LTsym_{X'} F = \LTsym_{X} F
    \label{e:LT4}
.
\end{align}
In particular, $\LTsym_{X}\circ (\Id - \LTsym_{X})=0$ on $\Ncal_X$.
\end{prop}

\begin{proof}
If $X=\varnothing$ then both sides are zero, so suppose that
$X,X'\not = \varnothing$.
Let $g \in \Phipol$.
By \eqref{e:LT2},
\begin{equation}
    \pair{\LTsym_{X}\circ \LTsym_{X'}F,g}_0
    =\pair{\LTsym_{X'}F,g}_0
    =
    \pair{F,g}_0 = \pair{\LTsym_{X}F,g}_0.
\end{equation}
Since $\LTsym_{X} \circ
\LTsym_{X'}F$ and $\LTsym_{X} F$ are both in $\Vcal (X)$, their
equality follows from the uniqueness in Definition~\ref{def:LTsym}.
\end{proof}

The following proposition gives an additivity property of $\LTsym$.

\begin{prop}
\label{prop:LTsum}
Let $X \subset \Lambda$ and $F_x \in \Ncal_{X}$
for all $x \in X$.  Suppose that $P \in \Vcal$
obeys $\LTsym_{\{x\}}F_x = P_x$ for
all $x\in X$.  Then $\LTsym_{X} F(X)=P(X)$, where $F(X)=\sum_{x\in
X}F_x$.
\end{prop}

\begin{proof}
If $X$ is empty then both sides are zero, so suppose that $X$ is not
empty.  Let $g \in \Phipol$.  It follows from \eqref{e:LT2}, linearity
of the pairing, and the assumption, that
\begin{align}
    \pair{\LTsym_XF(X),g}_0
&=
    \pair{F(X),g}_0 = \sum_{x\in X} \pair{F_x,g}_0
\\
&=
    \sum_{x\in X} \pair{\LTsym_{\{x\}}F_x,g}_0
=
    \sum_{x\in X} \pair{P_x,g}_0 = \pair{P(X),g}_0
.
\end{align}
Since $\LTsym_XF(X)$ and $P(X)$ are both in $\Vcal (X)$, their
equality follows from the uniqueness in Definition~\ref{def:LTsym}.
\end{proof}

For nonempty $X \subset \Lambda$, let $\Ecal(X)$ be the set of
automorphisms of $\Lambda$ which map $X$ to itself.  Here, an
\emph{automorphism} of $\Lambda$ means a bijective
map from $\Lambda$ to $\Lambda$
under which nearest-neighbour points are mapped to nearest-neighbour
points under both the map and its inverse.  In particular,
$\Ecal(\Lambda)$ is the set of automorphisms of $\Lambda$.  An
automorphism $E \in \Ecal(\Lambda)$ defines a mapping of the boson
field by $(\phi_E)_{x} = \phi_{Ex}$.  Then, for $F =\sum_{y \in
\vec\Lambdabold_f^*} \frac{1}{y!}F_y \psi^y \in \Ncal$, we define $E$
as a linear operator on $\Ncal$ by
\begin{align}
    \label{e:E-action-on-N}
    (EF)(\phi)
    &=
    \sum_{y\in \vec\Lambdabold_f^*} \frac{1}{y!}F_{y} (\phi_{E}) \psi^{Ey}
=
    \sum_{y\in \vec\Lambdabold_f^*} \frac{1}{y!}F_{E^{-1}y} (\phi_{E}) \psi^{y}
,
\end{align}
where in the second equality we have extended the action of $E$ to
component-wise action on $\Lambdabold_f$, and we used the fact that
summation over $y$ is the same as summation over $E^{-1}y$.  The
following proposition gives a Euclidean covariance property of
$\LTsym$.

\begin{prop}
\label{prop:9LTdef}
For $X \subset \Lambda$, $F \in \Ncal_{X}$
and $E \in \Ecal(\Lambda)$,
\begin{align}
    &
    E\big(\LTsym_{X} F\big) = \LTsym_{EX} (EF)
    \label{e:LT3}
.
\end{align}
\end{prop}

\begin{proof}
We define $E^{*}:\Phi \rightarrow \Phi$ by $(E^{*}g)_{z}=g_{Ez}$.
By \eqref{e:E-action-on-N}, and by taking derivatives with respect to
$\phi_{x_i}$ for $x_i \in \Lambdabold_b$, for $x \in \vec\Lambdabold_b^*$ we have
\begin{equation}
    \label{e:RFxy}
    (EF)_{x,y}(\phi) = F_{E^{-1}x,E^{-1}y} (\phi_{E}).
\end{equation}
Therefore,
\begin{align}
\label{e:Epair}
    \pair{E F,g}_\phi
    &= \sum_{z\in \vec\Lambdabold^*}
    \frac{1}{z!} F_{E^{-1}z} (\phi_{E}) g_{z}
    = \sum_{z\in \vec\Lambdabold^*}
    \frac{1}{z!} F_{z} (\phi_{E}) g_{Ez}
    = \pair{F,E^{*}g}_{\phi_{E}}
    .
\end{align}
Since $F \in \Ncal_{X}$ there exists a coordinate patch $\Lambda '$
containing $X$ such that $F \in \Ncal (\Lambda ')$. Let $g \in
\Phipol[E\Lambda ']$, and note that $E^{*}$ maps test functions in
$\Phipol[E\Lambda ']$ to test functions in $\Phipol[\Lambda ']$.  By
\eqref{e:LT2} and \eqref{e:Epair},
\begin{equation}
    \pair{E \LTsym_{X}F,g}_{0} = \pair{\LTsym_{X} F,E^{*}g}_{0}
    = \pair{F,E^{*}g}_{0} = \pair{E F,g}_{0}
    = \pair{\LTsym_{EX} E F,g}_{0}.
\end{equation}
Since both $E \LTsym_{X} F$ and $\LTsym_{EX} E F$ are in $\Vcal (EX)$,
their equality follows from the uniqueness in
Proposition~\ref{prop:LTsymexists}.
\end{proof}

The subgroup of $\Ecal(\Lambda)$ consisting of automorphisms that fix
the origin is homomorphic to the group $\Sigma$, with the element
$\Theta_E \in \Sigma$ determined from such an $E\in \Ecal(\Lambda)$ by
the action of $E$ on the set $\units$ of unit vectors.  Since
$\Ecal(\Lambda)$ is the semidirect product of the subgroup of
translations and the subgroup that fixes the origin, we can use this
homomorphism to associate to each element $E \in \Ecal(\Lambda)$ a
unique element $\Theta_E \in \Sigma$.
The following proposition ensures that the polynomial $P\in \Vcal$ determined
by $\LTsym_XF$ inherits symmetry properties of $X$ and $F$.

\begin{prop}
\label{prop:E-invariance}
For $X \subset \Lambda$ and $F \in \Ncal_{X}$ such that $EF=F$
for all $E\in\Ecal(X)$, the polynomial $P\in \Vcal$ determined by $P
(X)=\LTsym_{X} F \in\Vcal(X)$ obeys $\Theta_E P =
P$ for all $E\in\Ecal(X)$.
\end{prop}

\begin{proof}
By Proposition~\ref{prop:9LTdef} and by hypothesis, $EP(X) =
\LTsym_{EX}EF= P(X)$.
Therefore, for $g\in \Phipol$,
\begin{equation}
    \pair{F,g}_0 = \pair{P(X),g}_0 = \pair{EP(X),g}_0.
\end{equation}
Since $EP(X) = (\Theta_EP)(X)$, this gives
\begin{equation}
    \pair{P(X),g}_0 = \pair{(\Theta_EP)(X),g}_0,
\end{equation}
and since $\Theta_E P \in \Vcal$ by Proposition~\ref{prop:EVcal},
the uniqueness in Proposition~\ref{prop:LTsymexists} implies that
$\Theta_EP=P$, as required.
\end{proof}

The next two propositions concern norm estimates, using the $T_\phi$
semi-norm defined in \cite{BS-rg-norm}.  The $T_\phi$ semi-norm is
itself defined in terms of a norm on test functions, and next we
define the particular norm on test functions that we will use here.

The norm depends on a vector $\h = (\h_1,\ldots,
\h_{p_{\Lambdabold}})$ of positive real numbers, one for each field
species and component, where we assume that $\h_i$ depends only on the
field species of the index $k$.  Given $z=(z_1,\ldots, z_{p (z)}) \in
\Lambdabold^*$, we define $\h^{-z}= \prod_{i=1}^{p (z)}
\h_{k(z_i)}^{-1}$, where $k(z_i)$ denotes the copy of $\Lambda$
inhabited by $z_i\in \Lambdabold$.  Given $p_\Phi \ge 0$, the norm on
test functions is defined by
\begin{equation}
\label{e:Phignorm}
    \|g\|_{\Phi(\h)}  =
    \sup_{z\in \vec\Lambdabold^*}
    \sup_{|\alpha|_\infty \leq p_{\Phi}} \h^{-z}
    |\nabla_R^\alpha g_{z}| ,
\end{equation}
where $\nabla_R^\alpha = R^{|\alpha|_1}\nabla^\alpha$.  In terms of
this norm, a semi-norm on $\Ncal$ is defined by
\begin{equation}
    \|F\|_{T_\phi} = \sup_{g \in B(\Phi)}|\pair{F,g}_\phi|,
\end{equation}
where $B(\Phi)$ denotes the unit ball in $\Phi=\Phi(\h)$.  This
$T_\phi$ semi-norm depends on the boson field $\phi$, via the pairing
\eqref{e:pairdef}.

For the next two propositions, which provide essential norm estimates
on $\LTsym$, we restrict attention to the case where the torus
$\Lambda$ has period $L^N$ for integers $L,N>1$.  In practice, both
$L$ and $N$ are large.  We fix $j<N$ and take $R=L^j$.  The proofs of
the propositions, which make use of the results in
Section~\ref{sec:Taylor}, are deferred to
Section~\ref{sec:LTnormestimates}.  A $j$-\emph{polymer} is defined to
be a union of blocks of side $R=L^j$ in a paving of $\Lambda$.  Given
a $j$-polymer $X$, we define $X_+$ by replacing each block $B$ in $X$
by a larger cube $B_+$ centred on $B$ and with side $2L^j$ if $L^j$ is
even, or $2L^j-1$ if $L^j$ is odd (the parity consideration permits
centring).

\begin{prop}
\label{prop:Locbd} Let $L>1$, $j<N$, and let $X$ be a $j$-polymer with
$X_+ \subset U$ for a coordinate patch $U\subset \Lambda$.  For $F \in
\Ncal (U)$, there is a constant $\bar{C}'$, which depends only on
$L^{-j} {\rm diam}(U)$, such that
\begin{equation}
\label{e:LTs}
    \| \LTsym_{X} F\|_{T_0} \leq \bar{C}' \| F\|_{T_0}.
\end{equation}
\end{prop}

The next result, which is crucial, involves the $T_\phi$ semi-norm
defined in terms of $\Phi(\h)$, as well as the $T_\phi'$ semi-norm
defined in terms of the $\Phi'(\h')$ norm given by replacing $R=L^j$
and $\h$ of \eqref{e:Phignorm} by $R'=L^{j+1}$ and $\h'$.  In
addition, we assume that $\h'$ and $\h$ are chosen such that
$\h_i'/\h_i \le cL^{-[\phi_i]}$ for each component $i$, where $c$ is a
universal constant.  Let
\begin{equation}
\label{e:dplusprimedef}
    d_+' =
    \min \{ [M_m]  : m \not\in \mathfrak{v}_+ \},
\end{equation}
where $\mathfrak{v}_+$ was defined in \refeq{nufrak+};
thus $d_+'$ denotes the smallest dimension of a monomial not in
the range of $\LT$.
Let $[\varphi_{\rm min}] = \min\{[\varphi_i]: i =
1,\ldots,p_{\Lambdabold}\}$.
Given a
positive integer $A$, we define
\begin{equation}
\label{e:cgam}
    \cgam
    =
    L^{-d_{+}'} +  L^{-(A+1)[\varphi_{\rm min}]}
    .
\end{equation}
In anticipation of a hypothesis of Lemma~\ref{lem:phij}, for the next
proposition we impose the restriction that $p_\Phi \ge d_+'
-[\varphi_{\rm min}]$.  Its choice of large $L$ depends only on $d_+$.

\begin{prop}\label{prop:LTKbound}
Let $j<N$, let $A < p_\Ncal$ be a positive integer, let $L$ be
sufficiently large, let $X$ be a $j$-polymer with $X_+$ contained in a
coordinate patch, and let $Y \subset X$ be a nonempty $j$-polymer.
For $i=1,2$, let $F_{i} \in \Ncal(X)$.  Then
\begin{align}
    \|
     F_1(1-\LTsym_Y)F_2
    \|_{T_{\phi}'}
&\le
    \cgam \, \bar{C}
    \left(1 + \|\phi\|_{\Phi'}\right)^{A'} 
    \sup_{0\le t \le 1}
    \big(
    \|F_1F_2\|_{T_{t\phi}}
    +
    \|F_1\|_{T_{t\phi}}\|F_2\|_{T_{0}}\big),
\end{align}
where $\cgam$ is given by \eqref{e:cgam}, $A'=A+d_+/[\varphi_{\rm
min}]+1$, and $\bar{C}$ depends only on $L^{-j} {\rm diam}(X)$.
\end{prop}

For the special case with $F_1=1$, $F_2=F$, and $\phi = 0$,
Proposition~\ref{prop:LTKbound} asserts that
\begin{align}
\lbeq{1-LTex}
    \|
     F-\LTsym_X F
    \|_{T_{0}'}
&\le
    \cgam  \bar{C}
    \|F\|_{T_{0}}
    .
\end{align}
For the case of $d\ge 4$, $d_+=d$, $[\varphi_{\rm min}]=\frac{d-2}{2}$, and with
$A$ (and so $p_\Ncal$) chosen
sufficiently large that $(A+1)\frac{d-2}{2} \ge d +1$, we have
$d_+'=d_++1$ and
$\gamma=O(L^{-d-1})$.  This shows that, when measured
in the $T_0'$ semi-norm, $F-\LTsym_XF$ is substantially smaller than
$F$ measured in the $T_0$ semi-norm.

\subsection{An example}
\label{sec:locex}

The following example is not needed elsewhere in this paper, but it
serves to
illustrate the evaluation of $\LTsym$.

\begin{example}
\label{ex:LTmunu} Consider the case where there is a single complex
boson field $\phi$, in dimension $d=4$, with $[\varphi]=1$, and with
$d_+=d=4$.  The list of relevant and marginal monomials is as in
\eqref{e:locmonrel}--\eqref{e:locmonmar}, but now each factor of
$\varphi$ in those lists can be replaced by either $\phi$ or its
conjugate $\bar\phi$.  To define $\Vcal$, for each monomial $M$ we
choose $P(M)$ as in \eqref{e:Pm-def}, except monomials which contain
$\nabla^e\nabla^e$ for which we use $\nabla^{-e}\nabla^e$ as in
Example~\ref{ex:Pm-nonunique} instead.  Let $X\subset \Lambda$ be
contained in a coordinate patch and let $a,x \in X$.
\\
\emph{(i)}
Simple examples are given by
\begin{align}
    \LTsym_X
    |\phi_x|^6
    & = 0,
    \quad\quad
    \LTsym_{\{a\}}
    |\phi_x|^4
    = |\phi_{a}|^4,
\end{align}
which hold since in both cases the pairing requirement
of Definition~\ref{def:LTsym} is obeyed by the right-hand sides.
\\
\emph{(ii)}
Let $\tau_x = \phi_x\bar\phi_x$, let $q:\Lambda\to \C$, let $X$
be such that the range of $q$ plus the diameter of $X$ plus $2\bar
p_\Phi$ is strictly less than the period of the torus, and let
\begin{align}
    F
    &
    =
    \sum_{x \in X, y \in \Lambda} q (x-y) \tau_{y}
    .
\end{align}
The assumption on the range of $q$ ensures that the coordinate patch
condition in the definition of $\LTsym_XF$ is satisfied. We define
\begin{align}
    &
    q^{(1)}
    =
    \sum_{x \in \Lambda}q (x),
    \quad \quad
    q^{(**)}
    =
    \sum_{x \in \Lambda} q (x) x_{1}^{2},
    \label{e:bdef}
\end{align}
and assume that
\begin{align}
\label{e:qprop1}
    &
    \sum_{x \in \Lambda} q (x) x_{i} = 0,
    \quad\quad
    \sum_{x \in \Lambda} q (x) x_{i}x_{j}
    = q^{(**)} \delta_{i,j}
    \quad\quad\quad
    i,j \in \{1,2,\dotsc ,d \}.
\end{align}
We claim that
\begin{align}
\label{e:LTF3}
    \LTsym_{X} F
    &=
    \sum_{x\in X}
    \big(
    q^{(1)}\tau_{x}
    +
    q^{(**)} \sigma_{x}
    \big),
\end{align}
where, with $\Delta = - \sum_{i=1}^d \nabla^{-e_i}\nabla^{e_i}$,
\begin{equation}
    \label{e:tau-laplaciandefbis}
    \sigma_x =
    \frac{1}{2}
    \big(
    \phi_x \,\Delta \bar{\phi}_x  +
    \sum_{e\in \units}
    \nabla^{e}\phi_x \,\nabla^{e}\bar{\phi}_x  +
    \Delta\phi_x \,\bar{\phi}_x
    \big).
\end{equation}

To verify \eqref{e:LTF3}, we define
\begin{align}
    A
    &
    =
    \sum_{y \in \Lambda} q (a-y) \tau_{y}.
\end{align}
By Proposition~\ref{prop:LTsum}, it suffices to show that
\begin{align}
    \label{e:LTA3b}
    \LTsym_{\{a\}} A
    &=
    q^{(1)}\tau_{a}
    +
    q^{(**)} \sigma_{a}.
\end{align}
For this, it suffices to show that $A$ and $q^{(1)}\tau_{a} + q^{(**)}
\sigma_{a}$ have the same zero-field pairing with test functions $g
\in \Phipol$.  By definition, $\pair{A,g}_0 = \sum_{y\in \Lambda}
q(a-y) g_{y,y}$.  Since the polynomial test function $g =
g_{y_{1},y_{2}}$ is in $\Phipol$, it is a quadratic polynomial in
$y_{1},y_{2}$ and we can write the coefficients of this
polynomial in terms of lattice derivatives of $g$ at the point
$(a,a)$. For example the quadratic terms in $g$ are $(1/2)\sum_{i,j =
1}^d (y_i-a_i)(y_j-a_j) \nabla^{e_i}_1\nabla^{e_j}_2g_{a, a}$.  (The
construction of lattice Taylor polynomials is described below in
\eqref{e:Tay-def}.)

The constant term in $g$ is the zeroth derivative $g_{a,a}$.  The
linear terms vanish in the pairing due to \eqref{e:qprop1}.  For the
quadratic terms with derivatives on both variables of $g$, the only
non-vanishing contribution to the pairing arises from
$\frac 12 \sum_{i=1}^d(y_i-a_i)^2 \nabla^{e_i}_1\nabla^{e_i}_2g_{a,
a}$, due to \eqref{e:qprop1}, where the subscripts on the derivatives
indicate on which argument they act.  For the quadratic terms with
both derivatives on a single variable of $g$, by \eqref{e:qprop1} we
may assume that both derivatives are in the same direction, and for
those, we can replace the binomial coefficient $\binom{y_i-a_i}{2}$ by
$\frac 12 (y_i-a_i)^2$ due to the first assumption in
\eqref{e:qprop1}, to see that the relevant terms for the pairing are
\begin{equation}
    \frac 12 \sum_{i=1}^d (y_i-a_i)^2 \nabla^{e_i}_1\nabla^{e_i}_1 g_{a, a}
    +
    \frac 12 \sum_{i=1}^d ( y_i- a_i)^2
    \nabla^{e_i}_2\nabla^{ e_i}_2 g_{a, a}.
\end{equation}
Since $g$ is a  polynomial of total degree at most 2, we can use
\eqref{e:Vcal-relation} to replace derivatives $\nabla^e$ by $-\nabla^{-e}$
in the above expressions involving two derivatives.
Thus we obtain
\begin{equation}
    \pair{A,g}_0 =
    q^{(1)} g_{a, a} +
    q^{(**)} \frac 12
    \left(
    \Delta_1 g_{a, a}
    +
    \sum_{e \in \units} \nabla^e_1 \nabla^{e}_2 g_{a, a}
    + \Delta_2 g_{a, a}
    \right).
\end{equation}
By inspection, the right-hand side of \eqref{e:LTA3b}
has the same pairing with $g$ as $A$, so \eqref{e:LTA3b} is verified.
\\
\emph{(iii)}
Let
\begin{align}
    F '
    =
    \sum_{x \in X, y \in \Lambda} q (x-y)
    (
    \tau_{xy} + \tau_{yx}
    ).
\end{align}
By a similar analysis to that used in \emph{(ii)},
\begin{align}
\label{e:LTF4}
    \LTsym_{X} F '
    =
    \sum_{x\in X}
    \big(
    2 q^{(1)}\tau_{x}
    +
    q^{(**)} \frac{1}{2} \left(
    \phi_{x} \Delta \bar{\phi}_{x} +
    (\Delta \phi)_{x} \bar{\phi}_{x} \right)
    \big).
\end{align}
\qed
\end{example}

\subsection{Supersymmetry and \texorpdfstring{$\LTsym$}{loc}}
\label{sec:ssloc}

For our application to self-avoiding walk in \cite{BBS-saw4-log,BBS-saw4},
we will use $\LTsym$ in the context of a supersymmetric field
theory involving a complex boson field $\phi$ with conjugate $\bar\phi$,
and a pair of conjugate fermion fields $\psi,\bar\psi$, all of dimension
$\frac{d-2}{2}$.
We now show that if $F \in \Ncal$ is supersymmetric then so is $\LT_XF$.

The supersymmetry generator
${Q} = d + \ci$, which is discussed in \cite[Section~6]{BIS09}, has the
following properties: (i) $Q$ is an antiderivation that acts on
$\Ncal$, (ii) $Q^{2}$ is the generator of the gauge flow characterised by
$q \mapsto e^{-2\pi i t}q$ for $q = \phi_{x} ,\psi_{x}$
and
$\bar q \mapsto e^{+2\pi i t}\bar q$ for $\bar q = \bar\phi_{x} ,
\psib_{x}$, for all
$x \in \Lambda$.  An element $F \in \Ncal$ is said to be
\emph{gauge invariant} if it is invariant under this flow and
\emph{supersymmetric} if $QF=0$.  By property~(ii), supersymmetric elements
are gauge invariant.  Let $\hat{Q} = (2 \pi i)^{-1/2} Q$. Then
$\hat{Q}$ is an antiderivation satisfying:
\begin{align}
\label{e:Qaction}
    &
    \hat{Q}\phi = \psi,
    \quad \quad
    \hat{Q}\psi = - \phi,
    \quad\quad
    \hat{Q}\bar\phi = \psib,
    \quad \quad
    \hat{Q}\psib = \bar\phi
    .
\end{align}
The gauge flow clearly maps $\Vcal$ to itself.  Also, since the boson
and fermion fields have the same dimension, $Q$ also maps $\Vcal$ to
itself.  The following observation is a general one, but it has the
specific consequences that if $F$ is gauge invariant then so is
$\LTsym_{X}F$, and if $F$ is supersymmetric then
$Q\LTsym_{X}F=\LTsym_XQF=0$ so $\LTsym_{X}F$ is supersymmetric.  This
provides a simplifying feature in the analysis applied in
\cite{BS-rg-step}.

\begin{prop}
\label{prop:Qcom}
The map $Q: \Ncal \to \Ncal$ commutes with $\LTsym_{X}$.
\end{prop}

\begin{proof}
Let $F \in \Ncal$ and $g \in \Phipol$.  There is a map $Q^*:\Phipol
\to \Phipol$, which can be explicitly computed using
\eqref{e:Qaction}, such that $\pair{QF,g}_0=\pair{F,Q^*g}_0$. It then
follows from \eqref{e:LT2} that
\begin{equation}
    \pair{Q\LTsym_{X} F,g}_0
=
    \pair{\LTsym_{X} F,Q^{*}g}_0
=
    \pair{F,Q^{*}g}_0
=
    \pair{QF,g}_0
=
    \pair{\LTsym_{X}QF,g}_0
.
\end{equation}
Since $Q: \Vcal(X) \to \Vcal (X)$
by \eqref{e:Qaction}, it then
follows from the uniqueness in
Definition~\ref{def:LTsym} that $Q\LTsym_{X} F =
\LTsym_{X}QF$.
\end{proof}

The proof of Proposition~\ref{prop:Qcom} displays the general
principle that a linear map on $\Ncal$ commutes with $\LTsym_{X}$ if
its adjoint maps $\Pi$ to itself.  In particular, the map on $\Ncal$
induced by interchanging $\phi$ with its conjugate $\bar{\phi}$
commutes with $\LTsym_{X}$ for all $X$.

\subsection{Observables and the operator \texorpdfstring{$\LT$}{Loc}}
\label{sec:obsloc}

We now generalise the operator $\LTsym$ in two ways: to modify the set
onto which it localises, and to incorporate the effect of observable
fields.  The first of these is accomplished by the following
definition.

\begin{defn}
\label{def:LTsymXY} For $Y \subset X \subset \Lambda$ and $F \in
\Ncal_{X}$, we define the linear operator $\LTsym_{X,Y}:\Ncal
\to \Vcal(Y)$ by
\begin{align}
    \label{e:LTsymXYdef}
    &
    \LTsym_{X,Y}F = P_{X} (Y)
    \quad
    \text{with $P_{X}$ determined by $P_{X} (X) = \LTsym_{X}F$}
    .
\end{align}
\end{defn}

In other words, $\LTsym_{X,Y}F$ evaluates the polynomial $\LTsym_XF$
on the set $Y$ rather than on $X$.  It is an immediate consequence of
the definition that $\LTsym_X = \LTsym_{X,X}$, and that if
$\{X_1,\ldots, X_m\}$ is a partition of $X$ then
\begin{equation}
\label{e:LTsymXXi}
    \LTsym_{X}= \sum_{i=1}^m \LTsym_{X,X_{i}}.
\end{equation}
The following norm estimate for $\LTsym_{X,Y}$ is proved
in Section~\ref{sec:LTnormestimates}.

\begin{prop}
\label{prop:LTsymXYbd}
Suppose $\Lambda$ has period $L^N$ with $L,N>1$.  Let $j<N$, and let
$Y\subset X$ be $j$-polymers with $X_+ \subset U$ for a coordinate
patch $U\subset \Lambda$.  For $F \in \Ncal (U)$, there is a constant
$\bar{C}'$, which depends only on $L^{-j} {\rm diam}(U)$, such that
for $F \in \Ncal (U)$,
\begin{equation}
    \|\LTsym_{X,Y}F\|_{T_0} \le \bar{C}' \|F\|_{T_0}.
\end{equation}
\end{prop}

Next, we incorporate the presence of an \emph{observable field}.  The
observable field is not needed for our analysis of the self-avoiding
walk susceptibility in \cite{BBS-saw4-log}, but it is used in our
analysis of the two-point function in \cite{BBS-saw4}.  Specifically,
its application is seen in
\cite[Section~\ref{saw4-sec:of}]{BBS-saw4}. In that context we see
that the observable field $\sigma \in \C$ is a complex variable such
that differentiating the partition function with respect to $\sigma$
and $\bar{\sigma}$ at $\sigma =0$ gives the two-point function.  In
particular, elements of $\Ncal$ become functions of $\sigma$, and
given an element $F\in\Ncal$ we need the norm of $F$ to measure the
size of the derivatives of $F$ at zero with respect to $(\sigma
,\bar{\sigma})$. We can make our existing norm do this automatically
by declaring $(\sigma ,\bar{\sigma})$ to be a new species of complex
boson field, that is $\sigma$ is a function on $\Lambda$, but since we
do not need the additional information encoded by the dependence of
$\sigma$ on $x \in \Lambda$ we choose test functions that are constant
in $x$. This means that the norm only measures derivatives with
respect to observable fields that are constant on
$\Lambda$. Furthermore we choose test functions such that only
derivatives that are at most first order with respect to each of
$\sigma$ and $\bar{\sigma}$ are measured, since higher-order
dependence on $\sigma$ plays no role in the analysis of the two-point
function.

Thus, let $\sigma$ be a new species of complex boson field.  The norm
on test functions is defined as in \cite{BS-rg-norm}, with the
previously chosen weights $w_{\alpha_i,z_i}^{-1} =
\h_i^{-z_i}R^{|\alpha|}$ for the non-observable fields.  For the
observable field, we choose the weights differently, as follows.
First, if $\alpha \not = 0$ then we choose $w_{\alpha_i,z_i}=0$ when
$i$ corresponds to the observable species.  This eliminates test
functions which are not constant in the observable variables.  In
addition, we set test functions equal to zero if their observable
variables exceed one $\sigma$, one $\bar\sigma$, or one pair
$\sigma\bar\sigma$.  Therefore, modulo the ideal $\Ical$ of zero norm
elements, a general element $F \in \Ncal$ has the form
\begin{equation}
\label{e:1Fab}
    F
    =
    F^{\varnothing} +
    F^{\pp} +
    F^{\qq} +
    F^{\pp \qq}
,
\end{equation}
where $F^\varnothing$ is obtained from $F$ by setting $\sigma=\bar\sigma=0$,
while $F^{\pp} = F_\sigma \sigma$, $F^{\qq} = F_{\bar\sigma} \bar\sigma$,
and $F^{\pp \qq} =F_{\sigma ,\bar\sigma} \sigma\bar\sigma$ with
the derivatives evaluated at $\sigma=\bar\sigma=0$.
In the $T_{\phi}$ semi-norm we will \emph{always} set
$\sigmaa = \sigmab = 0$.
We unite the above cases with the notation $F^\alpha = F_\alpha \sigma^\alpha$
for $\alpha \in \{\varnothing, a,b,ab\}$.
This corresponds to a direct sum decomposition,
\begin{equation}
\label{e:1Ncaldecomp}
    \Ncal/\Ical
    =
    \Ncal^{\varnothing} \oplus
    \Ncal^{\pp} \oplus
    \Ncal^{\qq} \oplus
    \Ncal^{\pp\qq}
,
\end{equation}
with canonical
projections $\pi_\alpha : \Ncal/\Ical \to\Ncal^\alpha$ defined by
$\pi_\varnothing F = F_\varnothing$, $\pi_\pp F = F_\pp \sigma$, and
so on.
Note that
\begin{equation}
\label{e:Fnormsum}
    \|F\|_{T_\phi}=\sum_{\alpha}\|F_\alpha \sigma^\alpha\|_{T_\phi}
    =\sum_{\alpha}\|F_\alpha \|_{T_\phi}\| \sigma^\alpha\|_{T_0},
\end{equation}
by definition.  We use the same value $\h_\sigma$ in the weight
for both $\sigma$ and $\bar\sigma$.  In particular,
$\h_\sigma = \| \sigma \|_{T_0}=\| \bar\sigma \|_{T_0}$.

On each of the subspaces on the right-hand side of \eqref{e:1Ncaldecomp},
we choose a value for the parameter $d_+$ and construct corresponding
spaces $\Vcal^\varnothing, \Vcal^{\pp},\Vcal^{\qq} ,\Vcal^{\pp\qq}$
as in Definition~\ref{def:Vcal}.
We allow the freedom to choose different values for the
parameter $d_{+}$
in each subspace, and in our application in \cite{BBS-rg-pt,BS-rg-IE}
we will make use
of this freedom.  Then we define
\begin{equation}
\label{e:Vcalbigspace}
    \Vcal = \Vcal^{\varnothing} \oplus
    \Vcal^{\pp} \oplus
    \Vcal^{\qq} \oplus
    \Vcal^{\pp\qq}
.
\end{equation}

The following definition extends the definition of the localisation
operator by applying it in a graded fashion in the above direct sum
decomposition.

\begin{defn}
\label{def:LTXYsym}
Let $\Lambda '$ be a coordinate patch. Let $a,b
\in \Lambda'$ be fixed.  Let $X(\varnothing)=X$, $X(a) = X \cap
\{a\}$, $X(b)=X \cap \{b\}$, and $X(ab) = X \cap\{a,b\}$.  For
$Y\subset X \subset \Lambda'$ and $F \in \Ncal_{X}$, we define the
linear operator $\LT_{X,Y}:\Ncal_{X} \to \Vcal(Y)$ by specifying its
action on each subspace in \refeq{1Ncaldecomp} as
\begin{align}
    \label{e:LTXYdef}
    &
    \LT_{X,Y} F^\alpha  =
    \sigma^\alpha \LTsym_{X(\alpha),Y(\alpha)}^\alpha F_\alpha
    ,
\end{align}
and the linear map $\LT_X : \Ncal_{X} \to \Vcal(X)$ by
\begin{align}
    \label{e:LTdef2}
    &
    \LT_{X}F = \LT_{X,X}F =
    \LTsym_{X}^\varnothing F_{\varnothing} +
    \sigma \LTsym_{X(a)}^a F_{\pp} +
    \bar\sigma \LTsym_{X(b)}^b F_{\qq} +
    \sigma \bar\sigma\LTsym_{X(ab)}^{ab} F_{\pp \qq}
    .
\end{align}
The space $\Vcal$ is defined by \eqref{e:Vcalbigspace}.  Different
choices of $d_+$ are permitted on each subspace, and the label
$\alpha$ appearing on the operators $\LTsym$ on the right-hand side of
\eqref{e:LTXYdef}--\eqref{e:LTdef2} is present to reflect these
choices.  The use of $\Vcal(X)$ to denote the range of $\LT_X$ is a
convenient abuse of notation, which does not explicitly indicate that
the range on the four subspaces in the four terms on the right-hand
side of \refeq{LTdef2} are actually $\Vcal^\alpha(X(\alpha))$.
\end{defn}

It is immediate from the definition that
\begin{align}
    &
    \pi_{\alpha} \LT_{X,Y} = \LT_{X,Y} \pi_{\alpha}\quad
    \text{for $\alpha = \varnothing , \pp ,\qq , \pp \qq$},
\label{e:LTXY1}
\end{align}
and from \eqref{e:LTsymXXi} that, for a partition $\{X_1,\ldots, X_m\}$ of $X$,
\begin{align}
    &
    \LT_{X} = \sum_{i=1}^m \LT_{X,X_{i}}
    \label{e:LTXY6}
    .
\end{align}
It is a consequence of Proposition~\ref{prop:LT2} that
\begin{equation}
\label{e:oLT4}
    \LT_{X'}\circ \LT_X = \LT_{X'} \quad
    \text{for $X' \subset X \subset \Lambda$},
\end{equation}
and therefore
\begin{equation}
\label{e:oLT4b}
    \LT_{X}\circ (\Id - \LT_X) = 0
.
\end{equation}
Also, by Proposition~\ref{prop:9LTdef}, for an automorphism $E\in
\Ecal(\Lambda)$,
\begin{equation}
\label{e:oLT3}
    E\big(\LT_{X} F\big) = \LT_{EX} (EF) \quad\quad
    \text{if $F \in \Ncal^{\varnothing}_{X}$}.
\end{equation}
Note that \eqref{e:oLT3} fails in general for $F \in \Ncal_{X}
\setminus \Ncal^\varnothing_{X}$, due to the fixed points $a,b$ in
the definition of $\LT_{X,Y}F$.  The following two propositions extend
the norm estimates for $\LTsym$ to $\LT$.

\begin{prop}\label{prop:opLTdefXY}
Suppose $\Lambda$ has period $L^N$ with $L,N>1$.  Let $j<N$, and let
$Y\subset X$ be $j$-polymers with $X_+ \subset U$ for a coordinate
patch $U\subset \Lambda$.  For $F \in \Ncal (U)$, there is a constant
$\bar{C}'$, which depends only on $L^{-j} {\rm diam}(U)$, such that
for $F \in \Ncal (U)$,
\begin{equation}
\label{e:LTXY5}
    \|\LT_{X,Y}F\|_{T_0} \le \bar{C}' \|F\|_{T_0}.
\end{equation}
Note that the case $X=Y$ gives \eqref{e:LTXY5} for $\LT_XF$.
\end{prop}

\begin{proof}
By definition, the triangle inequality, Proposition~\ref{prop:LTsymXYbd},
and \refeq{Fnormsum},
\begin{align}
    \|\LT_{X,Y} F\|_{T_0}
    &=
    \sum_{\alpha = \varnothing, a,b, ab}
    \|\sigma^\alpha \LTsym_{X,Y}^\alpha F_\alpha\|_{T_0}
    \le
    \bar C' \sum_{\alpha = \varnothing, a,b, ab}
    \| \sigma^\alpha\|_{T_0}\| F_\alpha\|_{T_0}
    =
    \bar{C}' \|F\|_{T_0},
\end{align}
where $\bar C'= \max_\alpha \bar C_\alpha'$,
with $\bar C_\alpha'$ the constant arising in each of the four
applications of
Proposition~\ref{prop:LTsymXYbd}.
\end{proof}

For the next proposition, which is applied in
\cite[Proposition~\ref{IE-prop:1-LTdefXY}]{BS-rg-IE}, we write
$d_{\alpha}$ for the choice of $d_{+}$, and
$[\varphi_{\rm min}]$ for the common
minimal field dimension on each space $\Ncal^\alpha$ for
$\alpha = \varnothing , \pp ,\qq$ and $\pp \qq$.  We choose the spaces
$\Phi(\h)$ and $\Phi'(\h')$ as in Proposition~\ref{prop:LTKbound}.
With $d_\alpha'$
defined as in \refeq{dplusprimedef},
let
\begin{equation}
\label{e:cgamobs}
    \gamma_{\alpha,\beta}
        =
    (L^{-d_{\alpha}'} +  L^{-(A+1)[\varphi_{\rm min}]})
    \left( \frac{\h'_\sigma}{\h_\sigma} \right)^{|\alpha \cup \beta|}
    .
\end{equation}
As in Proposition~\ref{prop:LTKbound},
for the next proposition we again
require that $p_\Phi \ge d_+' -[\varphi_{\rm min}]$ and consider the case
where $\Lambda$ has period $L^N$.

\begin{prop}\label{prop:1-LTdefXY}
Let $j<N$, let $A < p_\Ncal$ be a positive integer, let $L$ be
sufficiently large, let $X$ be a $j$-polymer with enlargement $X_+$
contained in a coordinate patch, and let $Y \subset X$ be a nonempty
$L^j$-polymer.  For $i=1,2$, let $F_{i} \in \Ncal(X)$, with
$F_{2,\alpha}=0$ when $Y(\alpha)=\varnothing$.  Let $F =
F_1(1-\LT_{Y})F_2$.  Then
\begin{align}
\label{e:LTXY5a}
    \|F\|_{T_{\phi}'}
&\le
   \bar{C}
    \!\! \!\!
    \sum_{\alpha,\beta=\varnothing ,\pp ,\qq,\pp\qq}
    \cgam_{\alpha,\beta}
    \left(1 + \|\phi\|_{\Phi'}\right)^{A'}
    \nnb & \quad\quad\quad \times
    \sup_{0\le t \le 1}
    \big(
    \|F_{1,\beta}F_{2,\alpha}\|_{T_{t\phi}}
    +
    \|F_{1,\beta}\|_{T_{t\phi}}\|F_{2,\alpha}\|_{T_{0}}\big)
    \|\sigma^{\alpha\cup\beta}\|_{T_0},
\end{align}
where $\cgam$ is given by \eqref{e:cgam}, $A'=A+d_+/[\varphi_{\rm
min}]+1$, and $\bar{C}$ depends only on $L^{-j} {\rm diam}(X)$.
\end{prop}

\begin{proof}
We use
\begin{equation}
    \|F\|_{T_{\phi}'}
    \le \sum_{\alpha,\beta} \|\sigma^{\alpha\cup\beta} \|_{T_0'}
    \|F_{1,\beta}(1-\LTsym_{Y(\alpha)}^\alpha)F_{2,\alpha}\|_{T_\phi'}
\end{equation}
and apply Proposition~\ref{prop:LTKbound} to each term.
We also use
\begin{equation}
    \|\sigma^{\alpha\cup\beta} \|_{T_0'}
    =
    (\h_\sigma')^{|\alpha \cup\beta|}
    =
    \|\sigma^{\alpha\cup\beta} \|_{T_0}
    \left( \frac{\h'_\sigma}{\h_\sigma} \right)^{|\alpha \cup \beta|}
    .
\end{equation}
The constant $\bar C$ is the largest of the
four constants $\bar C_\alpha$ arising from
Proposition~\ref{prop:LTKbound}.
\end{proof}

\section{The operator \texorpdfstring{$\LTsym$}{loc}}
\label{sec:LTsym}

In Section~\ref{sec:LTsym-eu}, we prove existence of the operator
$\LTsym$ and prove Proposition~\ref{prop:LTsymexists}.  In
Section~\ref{sec:LTnormestimates}, we prove
Propositions~\ref{prop:Locbd}--\ref{prop:LTKbound}, using the results
on Taylor polynomials proven in Section~\ref{sec:Taylor}.  Finally, in
Section~\ref{sec:PM}, we now prove the claim which guaranteed
existence of the polynomials $\hat P$ used to define $\Vcal$ in
Definition~\ref{def:Vcal}.

Throughout this section, $\Lambda'$ is a coordinate patch in
$\Lambda$, and the space of polynomial test functions is $\Phipol =
\Phipol[\Lambda']$.

\subsection{Existence and uniqueness of \texorpdfstring{$\LTsym$}{loc}: Proof of
Proposition~\ref{prop:LTsymexists}}
\label{sec:LTsym-eu}

Recall from \cite[Proposition~\ref{norm-prop:pairingS}]{BS-rg-norm}
that the pairing obeys
\begin{equation}
\label{e:FSg}
    \pair{F,g}_\phi = \pair{F,Sg}_\phi
\end{equation}
for all $F \in \Ncal$, $g\in \Phi$, and for all boson fields $\phi$.
The symmetry operater $S$ is defined in
\cite[Definition~\ref{norm-def:S}]{BS-rg-norm}; it obeys $S^2=S$.  Let
$m \in \mathfrak{m}$ have components $m_k=(i_k,\alpha_k)$ for
$k=1,\ldots, p(m)$.
Recall that $m$ determines an abstract monomial $M_{m}$ by
\eqref{e:Mm} and, given $a \in \Lambda$, $M_{m}$ determines $M_{m,a}
\in \Ncal$ by evaluation of $M_{m}$ at $a$.  Recall from
\cite[Example~\ref{norm-ex:pairing}]{BS-rg-norm} that, for any test
function $g$,
\begin{equation}
\label{e:Mmg}
    \pair{M_{m,a}, g}_0
    =
    \nabla^{m} (Sg)_{\vec a},
\quad\quad
    \nabla^{m}
=
    \prod_{k=1}^{p (m)}
    \nabla^{\alpha_{k}},
\end{equation}
where on the right-hand side $\vec a$ indicates that each of the
$p(m)$ arguments is evaluated at $a$, and $\nabla^{\alpha_k}$ acts on
the variable $z_k$.

We specified a basis for $\Phipol$ in \eqref{e:pmdef}, but now we
require another basis. For $z=(x_1,\ldots,x_d)$ a coordinate on
$\Lambda'$, and $\alpha = (\alpha_{1},\dots ,\alpha_{d})\in
\Nbold_{0}^{d}$, we define the binomial coefficient $\binom{z}{\alpha}
= \binom{x_{1}}{\alpha_{1}} \dots \binom{x_{d}}{\alpha_{d}}$.  The new
basis is obtained by replacing, in the definition \eqref{e:pmdef} of
$p_{m}$, the monomial $z_{k}^{\alpha_{k}}$ by the polynomial
$\binom{z_{k}}{\alpha_{k}}$. More generally, we can also move the
origin. Thus for $m\in \bar{\mathfrak{m}}_+$ and $a \in \Lambda '$ we
define
\begin{equation}
\label{e:fmadef}
    b_{m,z}^{(a)}
=
    \prod_{k=1}^{p}\binom{z_{k}-a}{\alpha_{k}}
.
\end{equation}
This is a polynomial function defined on $\Lambda_{i_1}'^{(1)} \times
\cdots\times\Lambda_{i_{p(m)}}'^{(1)}$.  We implicitly extend it by
zero so that it is a test function defined on $\vec{\Lambdabold}^{*}$.
For $p(m)=0$, we set $b_\varnothing^{(a)}=1$.  For any $a \in \Lambda
'$, the set $\{b_{m,z}^{(a)}:m \in \bar{\mathfrak{v}}_{+} \}$ is a
basis for $\Phipol$.  For $g \in \Phi$, we define $\Tay_{a}: \Phi
\rightarrow \Phipol$ by
\begin{equation}
    \label{e:Tay-def}
    (\Tay_{a} g)_z
=
    \sum_{m \in \bar{\mathfrak{v}}_{+}}
    (\nabla^{m}g)_{\vec{a}} \, b_{m,z}^{(a)}
.
\end{equation}
The following lemma shows that $\Tay_{a} g$ is the lattice analogue of
a Taylor polynomial approximation to $g$.  Its proof is given in
Section~\ref{sec:Tay1}.

\begin{lemma}
\label{lem:Tay1}
Let $\Lambda'$ be a coordinate patch, and let $a,z\in\Lambda'$.
\\
(i) For $g\in \Phi$, $\Tay_{a}g$ is the unique $p \in \Phipol$ such
that $\nabla^{m} (g-p)_z|_{z=\vec{a}} = 0$ for all $m \in
\bar{\mathfrak{v}}_{+}$.
\\
(ii) $\Tay_{a}$ commutes with $S$.
\\
(iii) For $g \in \Phipol$, $(\Tay_{a}g)_z = g_z$.
\end{lemma}

For $m \in \mathfrak{m}_+$, let
\begin{equation}
\label{e:therealfmadef}
    f_{m}^{(a)}
    =
    N_m Sb_{m}^{(a)}
,
\end{equation}
where $N_{m}$ is a normalisation constant (whose value is chosen in
\refeq{Nmdef} so that case $m=m'$ holds in \eqref{e:dualbasis} below).
The lexicographic ordering on $\mathfrak{m}_+$ implies that
$f_{m}^{(a)} \not = f_{m'}^{(a)} \not = 0$ for $m\not =m'$.  Since
$\{b_{m}^{(a)}\}_{m\in \bar{\mathfrak{v}}_+}$ forms a basis of
$\Phipol$, the linearly independent set $\{f_{m}^{(a)}\}_{m\in
\mathfrak{v}_+}$ forms a basis of $S\Phipol$. The next lemma, which is
proved in Section~\ref{sec:dualpairing}, says that $\{M_{m,a}\}_{m\in
\mathfrak{v}_+}$ and $\{f_{m'}^{(a)}\}_{m'\in \mathfrak{v}_+}$ are
dual bases of $\Vcal_+$ and $S\Phipol$ with respect to the zero-field
pairing.

\begin{lemma}
\label{lem:Tay2}
Let $\Lambda'$ be a coordinate patch, and let $a,z\in\Lambda'$.
\\
(i) For $m,m' \in \mathfrak{m}_+$,
\begin{align}
    \label{e:dualbasis}
    \pair{M_{m,a},f_{m'}^{(a)}}_{0}
    &= \delta_{m,m'}.
\end{align}
(ii) For $g \in \Phi$,
\begin{equation}
    \label{e:Taydual}
    (\Tay_a S g)_z
    =
    \sum_{m\in \mathfrak{v}_+}
    \pair{M_{m,a},g}_0 f_{m,z}^{(a)}.
\end{equation}
\end{lemma}

\begin{defn}
\label{def:LTold}
Given $a \in \Lambda$, we define a linear map $\LTsym_{+,a} :
\Ncal_{\{a\}} \to \Vcal_{+}(\{a\})$ by
\begin{equation}
\label{e:LTolddef}
    \LTsym_{+,a} F
    =
    \sum_{m\in \mathfrak{v}_+} \pair{F,f_m^{(a)}}_0 M_{m,a}
.
\end{equation}
\end{defn}

It is an immediate consequence of \eqref{e:LTolddef} and
\eqref{e:dualbasis} that $\LTsym_{+,a} M_{m,a}=M_{m,a}$ for all $m\in
\mathfrak{v}_+$.  Since $\Vcal_+$ is spanned by the monomials
$(M_m)_{m\in \mathfrak{v}_+}$, it follows that
\begin{equation}
\label{e:LocPisP}
    \LTsym_{+,a} P_a = P_a \quad\quad P \in \Vcal_+.
\end{equation}
The following lemma shows that the map $\LTsym_{+,a}$ is dual to
$\Tay_a$ with respect to the zero-field pairing of $\Ncal$ and $\Phi$.

\begin{lemma}
\label{lem:Locid}
For any $a\in \Lambda$, $F \in \Ncal_{\{a\}}$, and $g\in \Phi$,
\begin{equation}
\label{e:Locid0}
    \langle \LTsym_{+,a}F, g \rangle_0 = \langle  F, \Tay_a g \rangle_0.
\end{equation}
In particular, if $g \in \Phipol$, then
\begin{equation}
\label{e:Locid}
    \langle \LTsym_{+,a} F, g \rangle_0 = \langle  F,  g \rangle_0.
\end{equation}
\end{lemma}

\begin{proof}
For \eqref{e:Locid0}, we use Definition~\ref{def:LTold},
linearity of the pairing, \eqref{e:Taydual}, Lemma~\ref{lem:Tay1}(ii)
and \eqref{e:FSg} to obtain
\begin{align}
    \langle \LTsym_{+,a} F, g \rangle_0
    & =
    \sum_{m \in \mathfrak{v}_+}
    \pair{F,f_m^{(a)}}_0
    \pair{M_{m,a},g}_0
    =
    \pair{F, \Tay_a S g}_0
    \nnb &
    =
    \pair{F, S \Tay_a g}_0
    =
    \pair{F, \Tay_a g}_0.
\end{align}
For \eqref{e:Locid}, we use \eqref{e:Locid0} and the
fact that $\Tay_a g = g$
for $g\in\Phipol$, by Lemma~\ref{lem:Tay1}(iii).
\end{proof}

\begin{lemma}
\label{lem:newLT} Let $a\in \Lambda$ and $X \subset \Lambda$ be such
that $X \cup \{a\}$ is contained in a coordinate patch.  Given $V_+
\in \Vcal_+$, there exists a unique $V \in \Vcal$ (depending on $V_+$,
$a$, and $X$) such that
\begin{equation}
    \LTsym_{+,a} V(X) = V_{+,a}.
\end{equation}
In particular, the map $V_{+} \mapsto V$ defines an isomorphism from
$\Vcal_+$ to $\Vcal$.
\end{lemma}

\begin{proof}
Fix $V_+ = \sum_{m \in \mathfrak{v}_+}\alpha_m M_{m,a}
\in \Vcal_{+}(\{a\})$; then $\alpha_m = \pair{V_{+,a},f_{m}^{(a)}}_0$
by \eqref{e:dualbasis}.
Let $\hat{P}_m = \hat{P}(M_m)$.
We want to show that there is a unique
$V  = \sum_{m'\in \mathfrak{v}_+} \beta_{m'} \hat P_{m'} \in \Vcal$ such
that
\begin{equation}
    \alpha_m =
    \sum_{m'\in \mathfrak{v}_+}\beta_{m'}
    \pair{\hat P_{m'}(X),f_m^{(a)}}_0
    =
    \sum_{m'\in \mathfrak{v}_+}\beta_{m'}
    B_{m',m},
\end{equation}
where $B_{m',m}=\pair{\hat P_{m'}(X),f_m^{(a)}}_0$.  Let $\hat
Q_{m'}=\hat{P}_{m'}-M_{m'}$.  According to Definition~\ref{def:Vcal},
$\hat Q_{m'} \in \Pcal_t + \Rcal_1$ for some $t > [M_{m'}]$.  By
definition, elements of $\Rcal_1(X)$ annihilate test functions in
pairings.  With \eqref{e:forlocplus}--\eqref{e:forRcal2} below, this
implies that, for $[M_{m'}]\ge [M_m]$,
\begin{equation}
    B_{m',m} =
    \pair{M_{m'}(X),f_m^{(a)}}_0
    +
    \pair{\hat Q_{m'}(X),f_m^{(a)}}_0
    =
    |X|\delta_{m',m} + 0 = \delta_{m',m}.
\end{equation}
Thus the matrix $B$ is triangular, with $|X|$ on the diagonal, and
hence $B^{-1}$ exists.  Then the row vector $\beta$ is given in terms
of the row vector $\alpha$ by $\beta = \alpha B^{-1}$, and this
solution is unique.  Since $\Vcal_+$ and $\Vcal$ have the same
finite dimension, the map $V_{+} \mapsto V$ defines an isomorphism
between these two spaces.
\end{proof}

The following commutative diagram illustrates the construction of
$\LTsym_X$ in the next proof:
\setlength{\unitlength}{.8mm}
\begin{center}
\begin{picture}(120,40)
\put(38,30){$\Ncal$}
\put(48,32){\vector(1,0){30}}
\put(57,36){$\LTsym_X$}
\put(80,30){$\Vcal(X)$}
\put(60,0){$\Vcal_{+}(\{a\})$}
\put(42,28){\vector(1,-1){21}}
\put(39,11){$\LTsym_{+,a}$}
\put(78,28){\vector(-1,-2){11}}
\put(79,14){$ \LTsym_{+,a}=\mu_{X,a}^{-1}$}
\end{picture}
\end{center}

\begin{proof}[Proof of Proposition~\ref{prop:LTsymexists}] (i)
\emph{Existence of $V \in \Vcal$}.  Given $a$ in $X$, let $\mu_{X,a}:
\Vcal_+(\{a\}) \to \Vcal(X)$ denote the map which associates the
polynomial $V(X)$ to $V_{+,a}$ in Lemma~\ref{lem:newLT}.  Let $V (X) =
(\mu_{X,a} \circ \LTsym_{+,a}) F$.  By \eqref{e:Locid} and
Lemma~\ref{lem:newLT}, for all $g \in \Phipol$,
\begin{equation}
    \pair{V (X),g}_0
    = \pair{\LTsym_{+,a} V (X),g}_0
    = \pair{\LTsym_{+,a} \mu_{X,a} \LTsym_{+,a} F,g}_0
    = \pair{\LTsym_{+,a} F,g}_0
    = \pair{F,g}_0.
\end{equation}
This establishes \eqref{e:LT2pair}.

\smallskip\noindent
(ii) \emph{Uniqueness}. Given two polynomials in $\Vcal$ that satisfy
\eqref{e:LT2pair}, let $P$ be their difference. Then $P$ is a
polynomial in $\Vcal$ such that, for all $g \in \Phipol$ and $a \in
X$,
\begin{equation}
    0
    = \pair{P (X),g}_{0}
    = \pair{\LTsym_{+,a} P (X),g}_{0}
,
\end{equation}
where we used \eqref{e:Locid}.  By \eqref{e:dualbasis}, $\LTsym_{+,a} P
(X)=0$ is zero as an element of $\Vcal_{+} (\{a \})$.  By
Lemma~\ref{lem:newLT}, $P=0$. This proves uniqueness.

\smallskip\noindent
(iii)
\emph{Independence of coordinate and coordinate patch}.
Recall the definition of $F \in \Ncal_X$ above
Proposition~\ref{prop:LTsymexists}. Suppose there are two coordinate
patches $\Lambda',\Lambda''$ with corresponding coordinates $z',z''$
that imply $F \in \Ncal_X$. Then there exists $V'$ such that
\eqref{e:LT2pair} holds for all $g \in \Phipol[\Lambda ']$ and $V''$
such that \eqref{e:LT2pair} holds for all $g \in \Phipol[\Lambda
'']$. In particular, $V'$ and $V''$ satisfy \eqref{e:LT2pair} for all
$g \in \Phipol[\Lambda ' \cap \Lambda '']$.  Since $\Lambda ' \cap
\Lambda ''$ with either of the coordinates $z',z''$ is also a valid
choice of coordinate patch that contains $X$, the uniqueness part (ii)
with coordinate patch $\Lambda ' \cap \Lambda ''$ implies $V'=V''$. So
the polynomial $V$ does not depend on the choice of $\Lambda'$
implicit in the requirement $F \in \Ncal_X$.

\smallskip\noindent
(iv) \emph{Duality}.
For $n \in \mathfrak{v}_+$, let $c_n$ be the
vector $(c_n)_{n'} = B^{-1}_{n,n'}$, where $B$ is the matrix
in the proof of Lemma~\ref{lem:newLT}.  It follows from that proof
that the pairing of
$\sum_{n'}(c_n)_{n'} \hat P_{n'}(X)$ with $f_m^{(a)}$ is
$\delta_{n,m}$.
Thus the basis $(c_n)_{n \in \mathfrak{v}_+}$
is dual to the basis $(f_m^{(a)})_{m \in \mathfrak{v}_+}$ of
$\Phipol$.
This completes the proof of
Proposition~\ref{prop:LTsymexists}.
\end{proof}

It follows from (i) and (ii) above that, for any $a \in X$,
\begin{equation}
\label{e:LTsymXdef}
    \LTsym_X F = (\mu_{X,a} \circ \LTsym_{+,a}) F,
\end{equation}

\subsection{Proof of norm estimates for \texorpdfstring{$\LTsym$}{loc}}
\label{sec:LTnormestimates}

We now prove
Propositions~\ref{prop:Locbd}, \ref{prop:LTKbound} and
\ref{prop:LTsymXYbd},
using the following definition which we recall from
\cite[\eqref{norm-e:PhiXdef}]{BS-rg-norm}. Given $X \subset
\Lambda$ and a test function $g \in \Phi$, we define
\begin{equation}
\label{e:PhiXdef}
    \|g\|_{\Phi(X)}
    =
    \inf \{ \|g -f\|_{\Phi} :
    \text{$f_{z} = 0$ if  all components of $z$ lie in $X$}\}.
\end{equation}
Let $f$ be as in \eqref{e:PhiXdef}. By definition, if $F \in \Ncal(X)$
then $\pair{F,g}_\phi = \pair{F,g - f}_\phi$.  Hence $ |\pair{F,g
}_{\phi}| \le \|F\|_{T_{\phi}}\, \|g - f\|_{\Phi }$, and by taking the
infimum over $f$ we obtain
\begin{equation}
\label{e:FXbd}
    |\pair{F,g}_\phi| \le  \|F\|_{T_{\phi}}\, \|g \|_{\Phi (X)}
    \quad \quad F\in \Ncal(X).
\end{equation}

\begin{proof}[Proof of Propositions~\ref{prop:Locbd} and
\ref{prop:LTsymXYbd}.]  We use the notation in the proof of
Lemma~\ref{lem:newLT}.  By definition, $\LTsym_{+,a}F = \sum_{m' \in
\mathfrak{v}_+} \alpha_{m'} M_{m',a}$ with $\alpha_{m'} =
\pair{F,f_{m'}^{(a)}}_0$.  Therefore, by \eqref{e:LTsymXdef} and the
formula $\beta = \alpha B^{-1}$ of the proof of Lemma~\ref{lem:newLT},
\begin{equation}
    \LTsym_X F = \sum_{m \in \mathfrak{v}_+} \beta_m \hat P_m(X)
    =
    \sum_{m,m' \in \mathfrak{v}_+} \pair{F,f_{m'}^{(a)}}_0
    B_{m',m}^{-1} \hat P_m(X).
\end{equation}
By Definition~\ref{def:LTsymXY}, this implies that
\begin{equation}
    \LTsym_{X,Y} F = \sum_{m \in \mathfrak{v}_+} \beta_m \hat P_m(Y)
    =
    \sum_{m,m' \in \mathfrak{v}_+} \pair{F,f_{m'}^{(a)}}_0
    B_{m',m}^{-1} \hat P_m(Y).
\end{equation}
Hence, writing $A=|X|^{-1}B$, and estimating the norm of $\hat P_m(Y)=
\sum_{y\in Y} \hat P_{m,y}$ by the triangle inequality, we obtain
\begin{align}
    \|\LTsym_{X,Y} F \|_{T_0}
    &\le
    \sum_{m,m' \in \mathfrak{v}_+} |\pair{F,f_{m'}^{(a)}}_0|\,
    |B_{m',m}^{-1}|\, \|\hat P_m(Y)\|_{T_0}
    \nnb & \le
    \frac{|Y|}{|X|}
    \sum_{m,m' \in \mathfrak{v}_+} |\pair{F,f_{m'}^{(a)}}_0|\,
    |A_{m',m}^{-1}|\, \|\hat P_{m,0}\|_{T_0}
    \nnb & \le
    \|F\|_{T_0}
    \frac{|Y|}{|X|}
    \sum_{m,m' \in \mathfrak{v}_+} \|f_{m'}^{(a)}\|_{\Phi(U)}
    \,
    |A_{m',m}^{-1}|\, \|\hat P_{m,0}\|_{T_0},
\end{align}
where we used \eqref{e:FXbd} in the last inequality.

It is shown in Lemmas~\ref{lem:Pmbdpf} and \ref{lem:TayX} that
\begin{equation}
\label{e:fPbds}
    \|\hat P_{m,0}\|_{T_0} \le R^{-|\alpha(m)|_1}\h^{m}
    ,
    \quad   \quad
    \|f_{m'}^{(a)}\|_{\Phi(U)} \le \bar C \h^{-m'} R^{|\alpha(m')|_1},
\end{equation}
where $\h^m$ denotes the product of $\h_{i_k}$ over the components
$(i_k,\alpha_k)$ of $m$.  It therefore suffices to show that
\begin{equation}
\label{e:Amm}
    |A_{m',m}^{-1}|
    \le
    \bar C \h^{m'} R^{-|\alpha(m')|_1}
    R^{|\alpha(m)|_1}\h^{-m}.
\end{equation}
The matrix elements $A_{m',m}$ can be computed using the formula
\begin{equation}
\label{e:Ainv}
    A_{m',m}^{-1} = \left( I + (A- I) \right)^{-1}
    = \sum_{j=0}^{|\mathfrak{v}_+|-1}(-1)^j (A- I)^j,
\end{equation}
where we have used the fact that the upper triangular matrix
$A- I$ with zero diagonal is nilpotent.  Consequently,
$A_{m',m}^{-1}$ is bounded by a sum of products of factors of the
form
\begin{equation}
    |X|^{-1}|\pair{\hat P_{m'}(X),f_m^{(a)}}_0|
    \le \|\hat P_{m',0}\|_{T_0} \|f_m^{(a)}\|_{\Phi(\hat X)},
\end{equation}
where $\hat X$ is a polymer which extends $X$ in a minimal way to
ensure that $P_{m'}(X) \in \Ncal(\hat X)$ for all $m'\in
\mathfrak{v}_+$. The extension is present because the discrete derivatives in
$P_{m'}$ cause $P_{m'}(X)$ to depend on points near the boundary, but
outside $X$. Now repeated application of \eqref{e:fPbds} gives rise to
a telescoping product in which the powers of $R$ and $\h$ exactly
cancel, leading to an upper bound
\begin{equation}
    \|\LTsym_{X,Y} F \|_{T_0}
    \le
    \bar C \|F\|_{T_0}.
\end{equation}
This proves Proposition~\ref{prop:LTsymXYbd},
and the special case $Y=X$ then gives Proposition~\ref{prop:Locbd}.
\end{proof}

For the proof of Proposition~\ref{prop:LTKbound}, we need some
preliminaries.  For $X$ contained in a coordinate patch $\Lambda'$,
let $\Phipol(X) \subset \Phi$ denote the set of test functions whose
restriction to every argument in $X$ agrees with the restriction of an
element of $\Phipol$.  This is not the same as $\Phipol[\Lambda']$
defined previously.)  Let
\begin{equation}
\lbeq{Phipolperpdef}
    \Phipol^{\perp} (X) = \{G\in \Ncal (X) :
    \pair{G,f}_{0}=0 \; \text{for all $f\in \Phipol(X)$}\}.
\end{equation}
We claim that $\Phipol^{\perp} (X)$ is an ideal in $\Ncal (X)$, namely that
\begin{equation}
    \label{e:Rcal-ideal}
    \pair{FG,f}_{0}=0
    \;\;
    \text{for all $F\in \Ncal (X)$, $G \in \Phipol^{\perp} (X)$, $f \in \Phipol(X)$}
.
\end{equation}

To prove \refeq{Rcal-ideal}, it suffices to consider test functions
$f\in \Phipol(X)$ which vanish except on sequences $z= (z_{1},\dots
,z_{p (z)})$ in $\vec\Lambdabold^{*}$ with $p (z)$ fixed equal to some
positive integer $n$.  Likewise, we can assume that $f_{z}=0$ unless
the component species $i (z_{1}),\dots ,i (z_{n})$ have specified
values.  These restrictions are sufficient because such test functions
span $\Phipol(X)$.  For such test functions, it follows from
\cite[\eqref{norm-e:Fdagf}]{BS-rg-norm} that $\pair{FG,f}_\phi =
\pair{G,F^\dagger f}_\phi$, where, for some constants $c_{z'}$,
\begin{align}
\label{e:Fdagf}
    (F^{\dagger}f)_{z''}
    &=
    \sum_{z'} c_{z'}F_{z'} \tilde f^{(z')}_{z''}
    \quad \text{with} \quad
     \tilde f^{(z')}_{z''} =     \sum_{z \in z' \diamond z''} f_z
.
\end{align}
For each fixed $z'$, the test function $\tilde f^{(z')}$ is an element
of $\Phipol(X)$, and hence $\pair{G,\tilde f^{(z')}}_0=0$.  Then
\eqref{e:Rcal-ideal} follows from \eqref{e:Fdagf} and the linearity of
the pairing.

We  define, on $\Phi$, the semi-norm
\begin{equation}
\label{e:Phitilnorm}
    \| g \|_{\tilde{\Phi}  (X)}
=
    \inf \{ \| g -f\|_{\Phi } : f \in \Phipol (X)\}
.
\end{equation}

\begin{lemma}
\label{lem:testfndecomp} Let $\epsilon >0$, $X\subset \Lambda'$, and
$g \in \Phi$.  Then there exists a decomposition $g=f+h$ with $f \in
\Phipol (X)$, $\|g\|_{\tilde{\Phi} (X)} \le \|h\|_{\Phi} \le
(1+\epsilon )\|g\|_{\tilde{\Phi} (X)}$ and $\|f \|_{\Phi} \le
(2+\epsilon) \|g\|_{\Phi}$.
\end{lemma}

\begin{proof}
By \eqref{e:Phitilnorm}, we can choose $f \in \Phipol (X)$ so that
$h=g-f$ obeys $\|g\|_{\tilde{\Phi}(X)} \leq \|h \|_{\Phi} \le
(1+\epsilon)\|g\|_{\tilde{\Phi}(X)}$, and then $\|f \|_{\Phi}\le
\|h\|_{\Phi} + \|g\|_{\Phi} \le (2+\epsilon) \|g\|_{\Phi}$.
\end{proof}

\begin{proof}[Proof of Proposition~\ref{prop:LTKbound}.]  Let $R=L^j$.
We write $c$ for a generic constant and $\bar c$ for a generic
constant that depends on $R^{-1}{\rm diam}(X)$.  Let $F \in \Ncal(X)$
and $A<p_\Ncal$.  We first apply
\cite[Proposition~\ref{norm-prop:Tphi-bound}]{BS-rg-norm} to obtain
\begin{align}
    \| F \|_{T_{\phi}'}
&\le
    \left(1 + \|\phi\|_{\Phi'}\right)^{A+1}
    \left[
    \|F\|_{T_{0}'}
     +
     \rho^{(A+1)} \sup_{0\leq t\leq 1} \|F\|_{T_{t\phi}} \right]
,
\end{align}
where, due to our choice of norm, $\rho^{(A+1)}\le c
L^{-(A+1)[\varphi_{\rm min}]}$.  To estimate $\|F\|_{T_{0}'}$, given a
test function $g$, we choose $f\in \Phipol(X)$ as in
Lemma~\ref{lem:testfndecomp}, and obtain
\begin{align}
\label{e:Fgf}
    \left| \pair{F,g}_{0}  \right|
&\le
    \left| \pair{F,f}_{0}\right| +
    \left| \pair{F,g-f}_{0}  \right|
.
\end{align}
Now we set $F=F_1(1-\LTsym_Y)F_2$.  By \eqref{e:LT2} and
\eqref{e:Phipolperpdef}, $(1-\LTsym_Y)F_2 \in \Phipol^{\perp} (X)$.
By \eqref{e:Rcal-ideal}, this implies that $F \in \Phipol^{\perp}
(X)$, so the first term on the right-hand side of \eqref{e:Fgf} is
zero.  For the second term, we use
\begin{align}
    \left| \pair{F,g-f}_{0}  \right|
&\le
    \|F\|_{T_{0} } \|g-f\|_{\Phi}
\le
    \|F\|_{T_{0}} (1+\epsilon)\|g\|_{\tilde{\Phi}}
\le
    \|F\|_{T_{0}} (1+\epsilon) \bar c
    L^{-d_+ '}\|g\|_{\Phi'}
,
\end{align}
where the final inequality is a consequence of Lemma~\ref{lem:phij}.
After taking the supremum over $g \in B (\Phi')$, followed by the
infimum over $\epsilon >0$, we obtain $ \|F\|_{T_{0} '}\le \bar
c\,
L^{-d_+'} \|F\|_{T_{0} }$, and hence
\begin{equation}
\label{e:cru1}
    \| F \|_{T_{\phi}'}
\le
    \left(1 + \|\phi\|_{\Phi'}\right)^{A+1}
    \bar c
    \,\left(
    L^{-d_+ '}
     +
    L^{-(A+1)[\varphi_{\rm min}]} \right)
    \sup_{0\leq t\leq 1} \|F\|_{T_{t\phi}}
.
\end{equation}

Next, we apply the triangle inequality and the product property of the
$T_\phi$ semi-norm to obtain
\begin{align}
    \|F \|_{T_{t\phi}}
&\le
    \|F_1F_2\|_{T_{t\phi}} +
    \|
    F_1
    \|_{T_{t\phi}}
    \|
    \LTsym_Y F_2
    \|_{T_{t\phi}}
.
\end{align}
Since $\LTsym_Y F_2 \in \Vcal$, it is a polynomial of dimension at most
$d_+$, and hence of degree at most
$d_+/[\varphi_{\rm min}]$. It follows from \cite[Proposition~\ref{norm-prop:T0K}]{BS-rg-norm} that
$\|\LTsym_Y F_2 \|_{T_{t\phi}}\le (1+\|\phi\|_\Phi)^{d_+/[\varphi_{\rm min}]}\|\LTsym_Y F_2 \|_{T_{0}}$.
With Proposition~\ref{prop:Locbd}, this gives
\begin{align}
\label{e:cru2}
    \|F \|_{T_{t\phi}}
&\le
    \|F_1F_2\|_{T_{t\phi}} +
    \bar C' (1+\|\phi\|_\Phi)^{d_+/[\varphi_{\rm min}]}
    \|
    F_1
    \|_{T_{t\phi}}
    \|
    F_2
    \|_{T_{0}}
.
\end{align}
Since $\|\phi\|_{\Phi} \leq cL^{-[\varphi_{\rm min}]} \|\phi\|_{\Phi'}
\leq c  \|\phi\|_{\Phi'}$
due to our choice of norm, this gives
\begin{align}
\label{e:cru3}
    \|F \|_{T_{t\phi}}
&\le
    \|F_1F_2\|_{T_{t\phi}} +
    \bar c (1+\|\phi\|_{\Phi'})^{d_+/[\varphi_{\rm min}]}
    \|
    F_1
    \|_{T_{t\phi}}
    \|
    F_2
    \|_{T_{0}}
.
\end{align}
Substitution of \eqref{e:cru3} into \eqref{e:cru1} completes the proof.
\end{proof}

\subsection{The polynomials \texorpdfstring{$P(M)$}{P (M)}}
\label{sec:PM}

We now prove the claim which guaranteed existence of the
polynomials $\hat P$ of Definition~\ref{def:Vcal}.  These polynomials
were used to define the
$\Sigma$-invariant subspace $\Vcal$ of $\Pcal$.

\begin{lemma}
\label{lem:Pm-exists} For any $M \in \Mcal_{+}$, the polynomial $P=P(M)$
of \eqref{e:Pm-def} obeys: (i) $P (M)$ is
$\Sigma_{\text{axes}}$-covariant, (ii) $M-P (M) \in \Pcal_{t} +
\Rcal_{1}$ for some $t>[M]$, and (iii) $P (\Theta M)=\Theta P (M)$
for $\Theta \in \Sigma_{+}$.
\end{lemma}

\begin{proof}
\emph{(i)} For $\Theta ' \in \Sigma_{\text{axes}}$,
\begin{align}
    \Theta 'P
&=
    |\Sigma_{\text{axes}}|^{-1}
    \sum_{\Theta  \in \Sigma_{\text{axes}}}
    \lambda (\Theta ,M) \Theta '\Theta M
\nnb
&=
    |\Sigma_{\text{axes}}|^{-1}
    \sum_{\Theta  \in \Sigma_{\text{axes}}}
    \lambda (\Theta '^{-1}\Theta ,M) \Theta M
\nnb
&=
    \lambda (\Theta '^{-1},M)
    |\Sigma_{\text{axes}}|^{-1}
    \sum_{\Theta  \in \Sigma_{\text{axes}}}
    \lambda (\Theta ,M) \Theta M
\nnb
&=
    \lambda (\Theta '^{-1},M) P
=
    \lambda (\Theta ',M) P
,
\end{align}
as required.
\\
\emph{(ii)}
Given $M\in \Mcal_{+}$ and $\Theta \in \Sigma_{\text{axes}}$, the monomial
$\Theta M$ is equal to $M$ with derivatives switched from forward to
backward in each coordinate where $\Theta$ changes sign.  Any
derivative that was switched can be restored to its original direction
using \eqref{e:Vcal-relation}, modulo a term in $\Pcal_t +
\Rcal_{1}$.  The use of \eqref{e:Vcal-relation} introduces a sign
change for each restored derivative, with the effect that $M$ is equal
to $\lambda (\Theta ,M) \Theta M$ modulo $\Pcal_t$.  Therefore,
$M-P(M)$ is also in $\Pcal_{t} + \Rcal_{1}$.
\\
\emph{(iii)} Let $M \in \Mcal_{+}$, $\Theta '\in \Sigma_{+}$, and
$\Theta \in \Sigma_{\text{axes}}$.  Since $\Theta'^{-1}\Theta \Theta'
\in \Sigma_{\text{axes}}$, it makes sense to write $\lambda (\Theta
'^{-1}\Theta \Theta ',M)$.  Also, since the number of derivatives that
change direction in the transformation $M \mapsto \Theta '^{-1}\Theta
\Theta 'M$ is equal to the number that change direction in the
transformation $\Theta'M \mapsto \Theta\Theta'M$, it follows that
$\lambda (\Theta '^{-1}\Theta \Theta ',M)= \lambda (\Theta ,\Theta
'M)$.  Therefore, by the change of variables $\Theta \mapsto
\Theta '^{-1}\Theta \Theta '$ in the sum,
\begin{align}
    \Theta 'P (M)
&=
    |\Sigma_{\text{axes}}|^{-1}
    \sum_{\Theta  \in \Sigma_{\text{axes}}}
    \lambda (\Theta ,M) \Theta '\Theta M
\nnb
&=
    |\Sigma_{\text{axes}}|^{-1}
    \sum_{\Theta  \in \Sigma_{\text{axes}}}
    \lambda (\Theta '^{-1}\Theta \Theta ',M) \Theta  \Theta 'M
\nnb
&=
    |\Sigma_{\text{axes}}|^{-1}
    \sum_{\Theta  \in \Sigma_{\text{axes}}}
    \lambda (\Theta ,\Theta 'M) \Theta  (\Theta 'M)
=
    P (\Theta 'M)
,
\end{align}
and the proof is complete.
\end{proof}

\section{Lattice Taylor polynomials}
\label{sec:Taylor}

\subsection{Taylor polynomials}
\label{sec:Tay1}

Let $\Lambda'$ be a coordinate patch, and let $a \in \Lambda'$.
Recall the definition of the test functions $b_{m}^{(a)}$ in
\eqref{e:fmadef}, for $m \in \bar{\mathfrak{m}}_+$.  We now prove
Lemma~\ref{lem:Tay1}.

\begin{proof}[Proof of Lemma~\ref{lem:Tay1}]
(i)
To show that $p=\Tay_a g$ obeys the desired identity
$\nabla^m(g-p)|_{z=\vec a} =0$, it suffices to show that
\begin{equation}
    \label{e:Tay1a}
    \nabla^{m}b_{m',z}^{(a)}|_{z=\vec a}=\delta_{m,m'},
    \quad\quad
    m,m' \in \bar{\mathfrak{m}}_{+}
.
\end{equation}
To prove \eqref{e:Tay1a},
it suffices to consider one species and the 1-dimensional case, since
the derivatives and binomial coefficients all factor.  For
non-negative integers $k,n$, it suffices to show that $\nabla_+^n
\binom{x-a}{k}|_{x=a} = \delta_{n,k}$, where we write $\nabla_+$ to
emphasise that this is a forward derivative.  We use induction on $n$,
noting first that when $n =0$ we have $\nabla_+^n
\binom{x-a}{k}|_{x=a} = \binom{0}{k} = \delta_{0,k} =\delta_{n, k}$.
To advance the induction, we assume that the identity holds for $n -1$
(for all $k\in \N_0$).  Since $\nabla_+ \binom{x-a}{k} =
\binom{x-a+1}{k}-\binom{x-a}{k} = \binom{x-a}{k -1}$ for all $x \in
\Z$, the induction hypothesis gives, as required,
\begin{equation}
\label{e:Dbin}
    \left. \nabla_+^n \binom{x-a}{k} \right|_{x=a}
    = \left. \nabla_+^{n-1} \binom{x-a}{k -1} \right|_{x=a}
    = \delta_{n -1,k -1} = \delta_{n,k}.
\end{equation}
For the uniqueness, suppose $q \in \Phipol$ obeys
$\nabla^m(g-q)|_{z=\vec a} =0$.  Since $\{b_{m}^{(a)}, m \in
\bar{\mathfrak{v}}_{+} \}$ is a basis of $\Phipol$, there are
constants $c_m$ such that $q=\sum_{m \in \bar{\mathfrak{v}}_+}c_m
b_{m}^{(a)}$.  By our assumption about $q$ and \refeq{Tay1a},
$\nabla^m g_{\vec a}=\nabla^m q_{\vec a}=c_m$, so $q=\Tay_a g$ as
required.
\\
(ii) It follows from \eqref{e:Tay-def} that the Taylor expansion of
$g$ with permuted arguments is obtained by permuting the arguments of
$\Tay_{a} g$, and from this it follows that $\Tay_{a}$ commutes
with $S$.
\\
(iii) This follows from the uniqueness in (i).
\end{proof}

We also make note of a simple fact that
we use below.  Suppose
the components of $m\in \bar{\mathfrak{m}}_+$ are $(i_{k},\alpha_{k})$
and the components of
$m'\in \bar{\mathfrak{m}}_+$ are $(i_{k},\alpha_{k}')$ where $k \in \{1,\dots ,p\}$ and
$\alpha_{k},\alpha_{k}' \in \Nbold_{0}^{d}$. We
say $\alpha_{k} \ge \alpha_{k}'$ if each component of $\alpha_{k}$ is
at least as large as the corresponding component of $\alpha_{k}'$. By
examining the proof of \eqref{e:Tay1a}, we find that
\begin{align}
    \label{e:Tay1}
    \nabla^{m}b_{m',z}^{(a)} & = 0
    \quad
    \text{if $\alpha_{k} > \alpha_{k}'$ for some $k = 1,\dots ,p$,}
    \\
    \label{e:nabb1}
    \nabla^{m}b_{m,z}^{(a)} & = 1.
\end{align}
In other
words, the condition $z = \vec{a}$ is not needed in these cases.

\subsection{Dual pairing}
\label{sec:dualpairing}

For $m \in \mathfrak{m}_+$ let $\vec\Sigma (m)$ be the set of
permutations of $1,\ldots, p (m)$ that fix the species when they act
on $m$ by permuting components, i.e., $\pi (i_{k},\alpha_{k}) =
(i_{\pi k},\alpha_{\pi k})$ with $i_{\pi k}= i_{k}$.   Let
$|\vec\Sigma (m) |$ be the order of this group.
There is also the subgroup
$\vec\Sigma_{0} (m)$ of permutations that fix $m$.  It has order
\begin{equation}
    |\vec\Sigma_{0} (m)|
    =
    \prod_{(i,\alpha)}n_{(i,\alpha)}(m)!
,
\end{equation}
with $n_{(i,\alpha)}$ as defined below \refeq{Mmex}: $n_{(i,\alpha)}$ denotes
the number of times that $(i,\alpha)$ appears as a component of $m$.

For example, for
$m=((1,\alpha_1),(1,\alpha_1),(1,\alpha_2),(1,\alpha_2),
(1,\alpha_2),(2,\alpha_3))$ with $\alpha_1 < \alpha_2$, we have
$|\vec\Sigma (m) |=5!1!$
and $|\vec\Sigma_{0} (m)|= 2!3!1!$. For this choice of $m$,
\begin{equation}
\label{e:bmex}
    b_{m,z}^{(a)}
    =
    \binom{z_1-a}{\alpha_1}\binom{z_2-a}{\alpha_1}
    \binom{z_3-a}{\alpha_2}\binom{z_4-a}{\alpha_2}
    \binom{z_5-a}{\alpha_2}\binom{z_6-a}{\alpha_3}.
\end{equation}
For this, or for any other
$m \in \bar{\mathfrak{m}}_{+}$, a permutation $\pi$
in $\vec\Sigma (m)$ has an action on $b_{m,z}^{(a)}$ either by mapping it to
$b_{\pi m,z}^{(a)}$ or to $b_{m,\pi z}^{(a)}$, where
$\pi (z_{1},\dots ,z_{p}) = (z_{\pi 1},\dots ,z_{\pi
p})$. The two actions are related by $b_{\pi m,z}^{(a)} =
b_{m,\pi^{-1}z}^{(a)}$.  Therefore $\vec\Sigma_{0} (m)$ is the set of
permutations that leave $b_{m,z}^{(a)}$ invariant.

By the definition of the symmetry operator $S: \Phi \to \Phi$ in
\cite[Definition~\ref{norm-def:S}]{BS-rg-norm}, for $m \in
\mathfrak{m}_{+}$,
\begin{equation}
    \label{e:Sb}
    (Sb_{m}^{(a)})_z
    =
    |\vec\Sigma (m) |^{-1}\sum_{\sigma \in \vec\Sigma (m)}
    \sgn(\sigma_f) b_{m,\sigma z}^{(a)},
\end{equation}
where
$\sigma_f$ denotes the restriction of $\sigma$ to the fermion
components of $z$, and $\sgn(\sigma_f)$ denotes the sign of this
permutation. In \eqref{e:therealfmadef}, we defined
\begin{equation}
\label{e:therealfmadefbis}
    f_{m}^{(a)}
    =
    N_m Sb_{m}^{(a)},
\end{equation}
and we now specify that
\begin{equation}
\label{e:Nmdef}
    N_m
    =
    \frac{|\vec\Sigma (m)|}
    {|\vec\Sigma_{0} (m) |}.
\end{equation}

We are now in a position to prove Lemma~\ref{lem:Tay2}.
Lemma~\ref{lem:Tay2}(i) is subsumed by Lemma~\ref{lem:dualbasis} and
is proved in \refeq{dualbasispf}.

\begin{proof}[Proof of Lemma~\ref{lem:Tay2}(ii)] Let $g \in \Phipol$.
By Lemma~\ref{lem:Tay1}(ii), $\Tay_{a} S = \Tay_{a} S^2 = S
\Tay_{a}S$.  With \eqref{e:Tay-def} and \refeq{Mmg}, this gives
\begin{align}
    \Tay_a (S g)
&=
    S\sum_{m \in \bar{\mathfrak{v}}_{+}}
    \pair{M_{m,a},g}_0 b_{m}^{(a)}.
\end{align}
Since $\vec\Sigma_0(m)$ is the set of permutations that leave
$m$ invariant, the sum over $\bar{\mathfrak{v}}_+$ can be written as
a sum over $\mathfrak{v}_+$, as
\begin{equation}
    S\sum_{m \in \bar{\mathfrak{v}}_{+}}
    \pair{M_{m,a},g}_0 b_{m}^{(a)}
    =
    S\sum_{m \in \mathfrak{v}_{+}}\frac{1}{|\vec\Sigma_0(m)|}
    \sum_{\sigma \in \vec\Sigma(m)}
    \pair{M_{\sigma m,a},g}_0 b_{\sigma m}^{(a)}.
\end{equation}
The anticommutativity of the fermions implies that
$\pair{M_{\sigma m,a},g}_0=\sgn(\sigma_f)\pair{M_{ m,a},g}_0$.
Since $b_{\sigma m,z}^{(a)} = b_{m,\sigma^{-1}z}^{(a)}$, it follows
from \refeq{Sb}--\refeq{Nmdef} and the fact that
$Sf_m^{(a)}=f_m^{(a)}$ that
\begin{equation}
    \Tay_a (S g)
    =
    S\sum_{m \in \mathfrak{v}_{+}}
    \pair{M_{m,a},g}_0 N_mSb_{ m}^{(a)}
    =
    S\sum_{m \in \mathfrak{v}_{+}}
    \pair{M_{m,a},g}_0 f_{ m}^{(a)}
    =
    \sum_{m \in \mathfrak{v}_{+}}
    \pair{M_{m,a},g}_0 f_{ m}^{(a)},
\end{equation}
and the proof is complete.
\end{proof}

The next lemma provides statements concerning the duality of field
monomials and test functions, for use in Section~\ref{sec:LTsym}.
In particular, \eqref{e:dualbasispf} gives Lemma~\ref{lem:Tay2}(i).

\begin{lemma}
\label{lem:dualbasis}
The following identities hold, for $a, x \in \Lambda'$:
\begin{align}
\label{e:dualbasispf}
    \pair{M_{m,a},f_{m'}^{(a)}}_{0} & = \delta_{m,m'}
    \quad\quad  m,m' \in \mathfrak{m}_+,
    \\
\label{e:forlocplus}
    \pair{M_{m,x},f_{m'}^{(a)}}_{0} & =
    \delta_{m,m'}
    \quad\quad
    m,m' \in \mathfrak{m}_+ \; \text{with $[M_m] = [M_{m'}]$},
\\
\label{e:forRcal2}
    \pair{M_{m,x}, f_{m'}^{(a)}}_0 & =0
    \quad\quad\quad\quad m \in \mathfrak{m} ,m' \in \mathfrak{m}_+\;\;
    \text{with $[M_m]>[M_{m'}]$}
.
\end{align}
\end{lemma}

\begin{proof}
We begin with a preliminary observation.
Let $m \in \mathfrak{m}$ and $m' \in \mathfrak{m}_+$.
It follows from \eqref{e:Mmg}, the identity $S^2=S$, and
\eqref{e:Sb}--\eqref{e:Nmdef} that
\begin{align}
\label{e:Mmgbis}
    \pair{M_{m,x}, f_{m'}^{(a)}}_0
    &= \nabla^{m} (Sf_{m'}^{(a)})|_{z={\vec x}}
    = |\vec\Sigma_{0} (m') |^{-1}
    \sum_{\sigma \in \vec\Sigma (m')}
    \sgn(\sigma_f)
    \nabla^{m}  b_{m',\sigma z}^{(a)}|_{z={\vec x}}
    \nnb
    &= |\vec\Sigma_{0} (m') |^{-1}
    \sum_{\sigma \in \vec\Sigma (m')}
    \sgn(\sigma_f)
    \nabla^{m}  b_{\sigma m',z}^{(a)}|_{z={\vec x}}
,
\end{align}
where for the last step we recall that
$b_{\pi m,z}^{(a)} =
b_{m,\pi^{-1}z}^{(a)}$.

It is now easy to prove \eqref{e:dualbasispf}.
Indeed, by \eqref{e:Tay1a} with $x=a$,
$\nabla^{m}  b_{\sigma m',z}^{(a)}|_{z={\vec a}}
=\delta_{m,\sigma m'}$.
For $m,m' \in \mathfrak{m}_{+}$, $m=\sigma m'$ holds if and only if
$m=m'$ and $\sigma \in \vec\Sigma_{0} (m')$. Since
$n_{(i,\alpha)}=1$ for fermion species $i$, we have
$\sgn(\sigma_f)=1$ for permutations that fix $m$, and
\eqref{e:dualbasispf} follows.

For the proof of \refeq{forlocplus}--\refeq{forRcal2}, we first
observe that by the definition of the zero-field pairing,
$M_{m,x}$ has nonzero pairing only with test functions
with the same number of variables as there are
fields in $M_{m,x}$.  Therefore, we may assume
that the number $p(m)$ of fields in $M_{m,x}$
is equal to the number
$p (m')$ of variables in $f_{m'}^{(a)}$.
Furthermore, the pairing only replaces the fields in $M_{m,x}$
with test functions whose arguments match the species of the fields.
Thus, for $m,m' \in \mathfrak{m}$,
the pairing $\pair{M_{m,x}, f_{m'}^{(a)}}_0$ is zero unless
$p(m)=p(m')$ and the
components $(i_{k},\alpha_{k})$ of $m$ and
the components $(i_{k}',\alpha_{k}')$ of $m'$
obey $i_k=i_k'$ for all $k=1,\ldots,p(m)$.
For \eqref{e:forlocplus},
the condition that $[M_m] = [M_{m'}]$ therefore becomes the condition that
$|\alpha|_1 = |\alpha'|_1$.  Consider first the case where $\alpha_{k}
\not = \alpha_{k}'$ for some $k$. Then, for some $k$,
$\alpha_{k}>\alpha_{k}'$. Since $m,m'$ are elements of
$\mathfrak{m}_{+}$ both the $\alpha_{k}$ and the
$\alpha_{k}'$ are ordered within each species.  Therefore it
is also true that for any permutation $\sigma \in \vec\Sigma (m')$
there is some $k$ such that $\alpha_{k}>\alpha_{\sigma
k}'$. By \eqref{e:Tay1}, in this case $\nabla^{m} b_{\sigma
m',z}^{(a)}=0$, so the right-hand side of \eqref{e:Mmgbis} is
zero. We are now reduced to the case $\alpha_{k}=\alpha_{k'}$ for all
$k$.  This means that $m=m'$ and we complete the proof of
\eqref{e:forlocplus} as in the proof of \refeq{dualbasispf},
applying \refeq{nabb1} rather than \refeq{Tay1a}.

Finally, we prove \eqref{e:forRcal2}.
As in the proof of \refeq{forlocplus},
the condition that $[M_m]>[M_{m'}]$ implies that for any $\sigma$
there is some $k$ such that $\alpha_{k}>\alpha_{\sigma k}'$. By
\eqref{e:Tay1}, this implies that $\nabla^{m} b_{\sigma m',z}^{(a)}=0$,
and hence the right-hand
side of \eqref{e:Mmgbis} is zero, and
\eqref{e:forRcal2} is proved.
\end{proof}

The following lemma is used in the proof of Proposition~\ref{prop:Locbd}.

\begin{lemma}
\label{lem:Pmbdpf}
For $m \in \mathfrak{v}_+$, let $\hat P_{m,x} = \hat P(M_{m,x})$,
with $\hat P$ given by Definition~\ref{def:Vcal}.  Then
there is a constant $c$ such that
\begin{equation}
\label{e:Pbdpf}
    \|\hat{P}_{m,x}\|_{T_0} \le  R^{-|\alpha(m)|_1} \h^{m},
\end{equation}
where $\h^m$ denotes the product of $\h_{i_k}$ over the components $(i_k,\alpha_k)$ of
$m$.
\end{lemma}

\begin{proof}
By Definition~\ref{def:Vcal}, $\hat P_m$ is a sum of monomials of the
same degree and dimension as $M_m$, so it suffices to prove
\eqref{e:Pbdpf} for a single such monomial $\tilde M_m$.  But for any
test function $g$, by \refeq{Mmg} and by the definition of the
$\Phi(\h)$ norm in \eqref{e:Phignorm}, we have
\begin{equation}
    |\pair{\tilde M_{m,x}, g}_0|
    =
    |\nabla^{\tilde\alpha(m)}(Sg)_z|_{z=\vec x}|
    \le
    R^{-|\alpha(m)|_1} \h^{m} \|Sg\|_{\Phi(\h)}
    \le
    R^{-|\alpha(m)|_1} \h^{m} \|g\|_{\Phi(\h)},
\end{equation}
as required.
\end{proof}

\subsection{Norm estimates and Taylor approximation}
\label{sec:Taybd}

The main results in this section are Lemmas~\ref{lem:TayX} and
\ref{lem:phij}, which are used in the proofs of
Propositions~\ref{prop:Locbd} and \ref{prop:LTKbound} respectively.
Lemma~\ref{lem:gX} is used to prove Lemmas~\ref{lem:TayX} and
\ref{lem:phij}, and Lemma~\ref{lem:taylor-theorem} is used to prove
Lemma~\ref{lem:phij}.  Lemmas~\ref{lem:gX}--\ref{lem:phij} are in
essence statements about test functions and Taylor approximation on
the infinite lattice $\Zd$, which we can apply to the torus $\Lambda$
by judicious restriction to a coordinate patch.  The correspondence
between $\Zd$ and a coordinate patch is possible since norms of test
functions are preserved by a coordinate $z$ as defined at the
beginning of Section~\ref{sec:oploc}, since nearest-neighbours and
hence derivatives are preserved by $z$.  Thus we work primarily in
this section on $\Zd$, with commentary in the statements of
Lemmas~\ref{lem:TayX} and \ref{lem:phij} concerning applicability on
the torus $\Lambda$.

Let $j<N$ and let $X$ be a $j$-polymer in $\Lambda$ or $\Zd$,
depending on context.  Recall that we defined an enlargement $X_+$ of
$X$ by doubling its blocks, above the statement of
Proposition~\ref{prop:Locbd}.  We extend this notion, as follows.  For
real $t>0$ and a nonempty $j$-polymer $X\subset \Zd$, let
$X_{t}\subset \Zd$ be the smallest subset that contains $X$ and all
points in $\Zd$ that are within distance $tL^{j}$ of $X$.  In
particular, $X_+=X_{1/2}$.  Below, we frequently write $R=L^j$.

The following lemma shows that, given $t>0$, it is possible to
estimate the $\Phi(X)$ norm of a test function $g$ using the values of
$g$ only in $X_{2t}$.  In its statement, we write $z \in {\mathbf
X}_{2t}$ to mean that each component $z_i$ of $z$ lies in $X_{2t}$.
Recall from \eqref{e:PhiXdef} that the $\Phi(X)$ norm is defined in
terms of the $\Phi=\Phi(\h)$ norm of \eqref{e:Phignorm} by
\begin{align}
\label{e:PhiXnormdef}
    \|g\|_{\Phi(X)}
    &=
    \inf \{ \|g -f\|_{\Phi} :
    \text{$f_{z} = 0$
    if  all components of $z$ lie in $X$}\},
\end{align}
where we can interpret $g$ as a test function either on $\Zd$ or
on $\Lambda$, depending on context.

\begin{lemma}
\label{lem:gX} Let $t>0$, $p \ge 1$, $j<N$, and let $X \subset \Zd$ be
a $j$-polymer.  There is a function $\chi_t$ of $p$ variables, which
takes the value $1$ if each variable lies in $X$, and the value $0$ if
any variable lies in $\Zd\setminus X_{2t}$, and a positive constant
$c_0$, independent of p, $X$ and $R=L^j$, such that for any
test function $g$ on $\Zd$ which depends on $p$ variables,
\begin{equation}
    \|g\|_{\Phi(X)}
    \le \|g\chi_t\|_{\Phi (\Zd)}
    \le
    \left((1+c_0t^{-1})\h^{-1}\right)^p
    \sup_{z \in {\mathbf X}_{2t}} \sup_{|\beta|_\infty \le p_\Phi}
    |\nabla^\beta_R g_{z}|.
\end{equation}
\end{lemma}

\begin{proof}
By definition, $g$ is a function of finite sequences each of
whose components is
in a disjoint union ${\mathbf X}$ of copies of $X$, where the
copies label species (fermions, bosons, field and conjugate field).
We give the proof for the special case
${\mathbf X}=X$, so that $g$
is a function of $z=(z_{1},\dots ,z_{p})$ with $z_i\in \Zd$. The
general proof is a straightforward elaboration of the notation.

Let $t>0$. We first construct a $t$-dependent function $\chi:\Rd
\rightarrow [0,1]$ such that
\begin{equation}
    \label{e:chi-properties}
    \chi\rvert_{X} = 1,
\quad\quad
    \chi\rvert_{\Zd \setminus X_{2t}} = 0,
\quad\quad
    \big|\nabla_{R}^{\alpha}\chi\rvert_{\Zd}  \big|
\le
    c (\alpha) t^{-|\alpha |_{1}}
,
\end{equation}
where $\nabla_{R}^{\alpha}=R^{|\alpha|_1}\nabla^\alpha$, and where the
estimate holds for all multi-indices $\alpha$ and is uniform in $X$.
Let $Y_{t}$ be the subset of $\Rd$ obtained by taking the union over
lattice points in $X_{t}$ of closed unit cubes centred on lattice
points.  Let $\varphi$ be a smooth non-negative function on $\Rd$
supported inside a ball of radius one and normalised so that $\int
\varphi dx = 1$. For $a = tR$, let $\varphi_{a} (x) = a^{-d}\varphi
(a^{-1}x)$ and let $\chi (x) = \int_{Y_{t}} \varphi_{a}(x-y)\,dy$. Then
\begin{equation}
    0
\le
    \chi (x)
\le
    \int_{\Rd}
    \varphi_{a} (x-y)
    \,dy
=
    \int_{\Rd}
    \varphi (x-y) \,dy
=
    1
\end{equation}
as required. For
$x\in X \subset \Rd$, the distance between $x$ and the
complement of $Y_{t}$ is at least $a$ and therefore $\chi (x) =
\int_{Y_{t}} \varphi_{a} (x-y)\,dy = \int_{\Rd } \varphi
(x-y)\,dy=1$. Therefore $\chi\rvert_{X} = 1$ as required. For $x \not
\in X_{2t}$, in the definition of $\chi$, $x-y$ is not in the support
of $\varphi_{a}$ so $\chi (x)=0$ as required.  The partial derivative
$\chi^{(\alpha)}$ of $\chi$ of total order $|\alpha |_{1}$ obeys
\begin{align}
    \big| \chi^{(\alpha)} (x)  \big|
&\le
    a^{-|\alpha |_{1}}
    \int_{X_{t}}
    \big|\varphi^{(\alpha)} (\frac{x-y}{a})\big|
    a^{-d}\,dy
\nnb & \le
    a^{-|\alpha |_{1}}
    \int_{\Rd}
    \big|\varphi^{(\alpha)} (x-y)\big| \,dy
\le
    c(\alpha) a^{-|\alpha |_{1}}
.
\end{align}
By the mean-value theorem, the finite difference derivative
$\nabla^{\alpha} \chi\rvert_{\Zd}$ is bounded by the continuum
derivative which is less than $c (\alpha) a^{-|\alpha |_{1}}$. When we
convert $\nabla$ derivatives to $\nabla_{R}$ derivatives the factors
of $R$ convert this estimate to $c (\alpha) t^{-|\alpha |_{1}}$ as
claimed.  This establishes the last estimate in
\eqref{e:chi-properties} and concludes the construction of $\chi$.

We extend $\chi$ to a function on sequences: for a sequence $z
=(z_1,\ldots,z_p)$, we define $\chi_{t}(z) = \prod_{i=1}^{p}
\chi(z_i)$.  Since $g\chi_t$ agrees with $g$ when evaluated on
${\mathbf X}$, and is zero outside ${\mathbf X}_{2t}$, it follows from
the definition of the $\Phi(X)$ norm in \eqref{e:PhiXdef} that
\begin{align}
    \|g\|_{\Phi(X)} \le \|g\chi_t\|_{\Phi(\Zd)}
    & \leq
    \sup_{z\in {\mathbf X}_{2t}} \h^{-z} \sup_{|\beta|_\infty  \le p_\Phi}
    | \nabla_R^\beta   (g\chi_t )_{z}   |.
\end{align}
Recall the lattice product rule $\nabla_{e} (hf)= (T_{e}f) \nabla h +
h\nabla f$ for differentiating a product, where $T_{e}$ is translation
by the unit vector $e$.  When the derivatives in $\nabla_{R}^{\beta}
(g\chi_t)$ are expanded using the lattice product rule, one of the
terms is $\chi_t \nabla_{k}^{\beta}g$.  The remaining terms all
involve derivatives of $\chi_t$, at most $p_\Phi$ in each coordinate.
This leads to a number of terms that grows exponentially in $p$, so
that, as required,
\begin{equation}
    \sup_{|\beta|_\infty  \le p_\Phi} | \nabla_R^\beta   (g\chi_t )_{z}   |
    \le
    \big(1+O(t^{-1})\big)^p\sup_{|\beta|_\infty  \le p_\Phi} | \nabla_R^\beta   g_{z}   |.
\end{equation}
This completes the proof.
\end{proof}

\begin{lemma}
\label{lem:TayX} Let $j<N$, let $m \in \mathfrak{m}_+$, let $X$ be a
$j$-polymer in $\Zd$, and let $a \in X$.  There is a constant
$\bar{C}$, independent of $m$ but dependent on the diameter of
$R^{-1}X$, such that for the polynomial $f_m^{(a)}$ defined on all of
$\Zd$,
\begin{equation}
    \label{e:Tay-bound}
    \|f_m^{(a)}\|_{\Phi (X)} \le \bar{C} \h^{-m} R^{|\alpha(m)|_1}.
\end{equation}
The same inequality holds for $f_m^{(a)}$ as we have defined it on the
torus, provided $X_+$ lies in a coordinate patch.
\end{lemma}

\begin{proof}
For the case of $\Zd$, by the definition of $f_m^{(a)}$ in
\eqref{e:therealfmadefbis}, and by Lemma~\ref{lem:gX} with $t=\frac
12$, it suffices to show that for $z \in {\mathbf X}_+$ and for
$|\beta|_\infty \le p_\Phi$,
\begin{equation}
\label{e:nabf}
    |\nabla_R^\beta b_{m,z}^{(a)}|
    \leq
    \bar{c} R^{|\alpha|_1},
\end{equation}
where $\bar{c}$ depends on $m$ and $R^{-1} X$. Note that any
dependence on $p$ (from Lemma~\ref{lem:gX}) and $m$ is uniformly
bounded since the number of variables in bounded when
$m \in \mathfrak{m}_{+}$.

To prove \refeq{nabf}, we first note that if any component of $\beta$
exceeds the corresponding component of $\alpha = \alpha(m)$ then the
left-hand side of \eqref{e:nabf} is equal to zero as in the proof of
\eqref{e:forRcal2}.  Thus we may assume that each component of $\beta$
is at most the corresponding component of $\alpha$, and without loss
of generality we may consider the 1-dimensional case.  In this case,
for $j=j_-+j_+ \le k$, $|\nabla_-^{j_-}\nabla_+^{j_+} \binom{x-a}{k}|
= |\binom{x-a-j_-}{k-j}|$ and this is at most a multiple of $R^{k-j}$,
with the multiple dependent on the ratio of the diameter of $X$ to
$R$.  This proves \eqref{e:nabf} and completes the proof of
\eqref{e:Tay-bound} for $\Zd$.  There is no dependence of $\bar
C$ on $m \in \mathfrak{m}_+$, since $\mathfrak{m}_+$ is a finite set.

This then implies the extension to the torus, since derivatives of
$b_m^{(a)}$ are the same on a coordinate patch and its image rectangle
in $\Zd$.
\end{proof}

The following Taylor remainder estimate is used to prove
Lemma~\ref{lem:phij}, which plays an important role in the proof of
the crucial change of scale bound in Proposition~\ref{prop:LTKbound}.
For its statement, given $a\in \Z^d$, $p \in \N$, $z=(z_1,\ldots,
z_p)$ with $z_1,\ldots,z_p \in \Z^d$ and with $(z_i)_j \ge a_j$ for
all $i=1,\ldots, p$ and $j=1,\ldots,d$, and $t \in \N$, we define
$S_t(a,z) = \{ y =(y_1,\ldots,y_p): y_i\in \Z^d : a_j -t \le (y_i)_j
\le (z_i)_j\}$.  We make use of the map $\Tay_{a} :\Phi \to \Phipol$
given by \eqref{e:Tay-def}, interpreted as a map on test functions $g$
defined on $\Zd$.  The range of $\Tay_a$ involves polynomials in the
components of $z$ to maximal degree $s = d_+ - \sum_{k=1}^p
[\varphi_{i(z_k)}]$, where $i(z_k)$ denotes the field species
corresponding to the component $z_k$.  Also, given a test function
$g\in \Phi^{(p)}$, we write $M_g = \sup_{y \in S_{s}(a,z)}
\sup_{|\alpha|_\infty =s+1} |\nabla^\alpha g_y|$ where the
supremum over $\alpha$ is a supremum over only \emph{forward}
derivatives.

\begin{lemma}\label{lem:taylor-theorem}
For $a \in \Zd$, components of $z=(z_1,\ldots, z_p)$ in $\Zd$ with
$(z_i)_j \ge a_j$ for all $i,j$, and for $|\beta|_{1}=t\le s$ (forward
or backward derivatives), the remainder in the approximation of
$g=g_z$ by its Taylor polynomial obeys
\begin{equation}
\label{e:Tayrem2}
    |\nabla^\beta (g-\Tay_a g)_z|
    \le
    M_g \binom{|z-\vec{a}|_{1} }{s-t+1} ,
\end{equation}
with $M_g$ and $s$ as defined above.
\end{lemma}

\begin{proof}
The proof is by induction on the dimension of $z \in \Z^{dp}$
and does not depend on the grouping of these components of $z$ into
$\Zd$.  Therefore we give the proof for case $d=1$.
Also without loss of generality, we assume that $a=0$.  Let $f_z =
\Tay_a g_z = \Tay_0 g_z$.

We first show that it suffices to establish \eqref{e:Tayrem2} for the
case $|\beta|_{1}=t=0$, namely
\begin{equation}
\label{e:Tayrem}
|g_z-f_z|\le M_g \binom{|z|_{1}}{s+1},
\end{equation}
with the supremum defining $M$ taken over $S_0(z)$.  In fact, for the
case where $\beta$ involves only forward derivatives, $\nabla^\beta f$
is the degree $s-t$ Taylor polynomial for $\nabla^\beta g$, and it
follows from \eqref{e:Tayrem} that
\begin{equation}
\label{e:Tayrem1}
|\nabla^\beta (g-f)_z|\le M_{g }\binom{|z|_{1}  }{ s-t+1} ,
\end{equation}
which is better than \eqref{e:Tayrem2}.  To allow also backward
derivatives, we simply note that a single backward derivative is equal
in absolute value to a forward derivative at a point translated
backwards, and this translation is handled in our estimate by the
extension of $S_0(z)$ to $S_t(z)$ in the definition of $M_g$.

It remains to prove \eqref{e:Tayrem}.  The proof is by induction on
$p$ (with $s$ held fixed).
Consider first the case $p=1$.  For a function
$\phi$ on $\Z$, let $(T\phi)_x =\phi_{x+1}$ and let $D=T-I$.
For $m >0$, $T^{m} = I +
\sum_{n= 1}^m (T-I)T^{n-1}$.
Iteration of this formula $s$ times gives
\begin{align}
    T^{m}
&=
    I + \sum_{m\ge n_{1}\ge 1}D I +
    \sum_{m\ge n_{1}>n_{2}\ge 1}D^{2}T^{n_{2}-1}
=
    \dotsb
=
    \sum_{\alpha=0}^{s}\binom{m}{\alpha}D^{\alpha} + E,
\end{align}
where
\begin{equation}
    E
=
    \sum_{m\ge n_{1} > n_{2} > \dotsb > n_{s+1} \ge 1}
    D^{s+1}T^{n_{s+1}-1}.
\end{equation}
We apply this operator identity to $( T^{z_1}g)_0$ and obtain, for $p=1$,
\begin{equation}
    g_{z_1} = (T^{z_1}g)_0 = f_{z_1} + (Eg)_0.
\end{equation}
The remainder term obeys the estimate
\begin{align}
    |(Eg)_0|
&\le
    \sum_{m\ge n_{1} > n_{2} > \dotsb n_{s+1} \ge 1} \
    \sup_{x\in S_0 (z_1)}
    |D^{s+1}g_x |
    =
    \binom{m }{ s+1} \sup_{x\in S_0 (z_1)}
    |D^{s+1}g_x |
    .
\end{align}
This proves \eqref{e:Tayrem} for $p=1$.

To advance the induction, we assume that \eqref{e:Tayrem} holds for
$p-1$.  We write $y = (z_1,\ldots, z_{p-1})$ and
$z=(y, z_{p})$, and apply the case $p-1$ to $g$ with the
coordinate $z_{p}$ regarded as a parameter.  This gives
\begin{equation}
\label{e:Td}
    g_z = \sum_{|\beta|_{1} \leq s} \binom{y}{ \beta}
    D^{ \beta} g_{(  0, z_{p})} + \tilde E,
\end{equation}
where by the induction hypothesis $|\tilde E |
\leq M  \binom{|y|_{1}}{s+1}$.  We also apply the case $p=1$ to
obtain
\begin{equation}
\label{e:T1}
    D^{ \beta} g_{(  0, z_{p})}
    =
    \sum_{\alpha = 0}^{s-|\beta|_{1}}
    \binom{z_{p}}{\alpha}
    D^\alpha D^{\beta} g_0 + E_1,
\end{equation}
with $|E_1| \leq M \binom{ z_{p} }{s-|\beta|_{1}+1}$.
The insertion of \eqref{e:T1} into \eqref{e:Td} yields
\begin{equation}
    g_z = \sum_{|\beta|_{1} \leq s} \binom{y }{ \beta}
    \sum_{\alpha = 0}^{s-|\beta|_{1}}
    \binom{z_{p} }{ \alpha}
    D^\alpha D^{\beta} g_0
    + \sum_{|\beta|_{1} \leq s} \binom{y}{\beta}E_1 + \tilde E.
\end{equation}
The first term on the right-hand side is just the Taylor polynomial
$f_z$ for $g_z$.
It therefore suffices to show that
\begin{equation}
\label{e:Vandermonde}
    \sum_{|\beta|_{1} \leq s} \binom{y}{\beta}
    \binom{ z_{p} }{s-|\beta|_{1}+1}
    +
    \binom{|y|_{1}}{s+1}
    \le
    \binom{|z|_1}{s+1}.
\end{equation}
However, \eqref{e:Vandermonde} follows from a simple counting argument:
the right-hand side counts the number of ways to choose $s+1$ objects
from $|z|_1$, while the left-hand side decomposes this into two terms,
in the first of which at least one object is chosen from the last coordinate
of $z$, and in the second of which no object is chosen from the last
coordinate.
This completes the proof of \eqref{e:Tayrem}.
\end{proof}

The following lemma is used in this paper in the proof of
Proposition~\ref{prop:LTKbound}, and it is also used in
\cite[Lemma~\ref{IE-lem:mart}]{BS-rg-IE}.  Its most natural setting is
$\Zd$, but we do require it in the case of a torus $\Lambda$ with
period $L^N$ for integers $L,N>1$.  Given $j<N$, let $R=L^j$ and
$R'=L^{j+1}$.  Let $\Phi(\h),\Phi'(\h')$ be test function spaces
defined via weights involving parameters $R=L^j,\h$ and
$R'=L^{j+1},\h'$ respectively.  Suppose that $\h'_i/\h_i \le
cL^{-[\phi_i]}$, where $c$ is a universal constant.

\begin{lemma}
\label{lem:phij} Suppose that $p_\Phi \ge d_+'-[\varphi_{\rm min}]$.
Fix $L >1$.  Let $j<N$ and let $X$ be an $L^j$-polymer on $\Zd$ with
enlargement $X_+$ as in Lemma~\ref{lem:gX} with $t=\frac 12$.  There
exists $\bar C_3$, which is independent of $L$ and depends on $j$ only
via $L^{-j}{\rm diam}(X)$, such that for any test function $g$ on
$\Zd$,
\begin{equation}
\label{e:phij}
    \|g\|_{\tilde{\Phi} (X)}
    \le
    \bar C_3
    L^{-d_{+}'}   \|g\|_{\tilde{\Phi}' (X_+)}
,
\end{equation}
with $d_+'$ given by \refeq{dplusprimedef}.  In particular,
$\|g\|_{\tilde{\Phi} (X)} \le \bar C_3 L^{-d_+ '} \|g\|_{\Phi'}$.  The
bound \refeq{phij} also holds for a test function $g$ on the torus
$\Lambda$, provided $L$ is sufficiently large and there is a
coordinate patch $\Lambda' \supset X_+$.
\end{lemma}

\begin{proof}
We first consider the case of $\Zd$.  We assume that $X$ is connected;
if it is not then the following argument can be applied in a
componentwise fashion.  For connected $X$, let $a$ be the largest
point which is lexicographically no larger than any point in $X$.

Given $g$, we use Lemma~\ref{lem:testfndecomp} to choose $f \in
\Phipol(X)$ such that $h = g -f$ obeys $\|h\|_{\Phi'(X)} \le 2
\|g\|_{\tilde{\Phi}' (X)}$.  Then $g-(h-\Tay_a h) \in \Phipol(X)$, and
hence
\begin{equation}
\label{e:ghTay}
    \|g\|_{\tilde\Phi(X)} = \| h - \Tay_a h \|_{\tilde\Phi(X)}
    \leq \| h - \Tay_a h \|_{\Phi (X)}.
\end{equation}
It suffices to prove that for every test function $h$,
\begin{align}
\label{e:phijnew}
    \|h - \Tay_a h\|_{\Phi (X)}
    & \le
    \frac 12 \bar C_3
    L^{-d_+ '} \|h \|_{\Phi' ( X_+)},
\end{align}
since $\|h\|_{\Phi'( X_+)}
    \leq
    2   \|g\|_{\tilde\Phi'( X_+)}
    \le 2\|g\|_{\Phi'}$.

The rest of the proof is concerned with proving \eqref{e:phijnew}.  We
write $R=L^j$ and $R'=L^{j+1}$.  Let $r= h - \Tay_a h$.  By
Lemma~\ref{lem:gX} with $t=\frac 12$, there is a constant $K>1$ such
that
\begin{align}
\label{e:gTay1}
    \|r\|_{\Phi (X)}
    & \le
    \sup_{z \in {\mathbf X}_+} (K\h^{-1})^z
    \sup_{|\beta|_\infty \le p_\Phi}
    | \nabla_R^\beta  r_{z}  |.
\end{align}
By the hypothesis on $\h'$, \eqref{e:gTay1} implies that
\begin{align}
\label{e:rhognew}
    \|r\|_{\Phi (X)}
    & \le
    \sup_{z \in {\mathbf X}_+}
    (cK\h'^{-1})^z
    \sup_{|\beta|_\infty \le p_\Phi}
    L^{-(\sum_k [\varphi_{i_k}]+|\beta|_1)}
    | \nabla_{R'}^\beta  r_{z}  |,
\end{align}
where the sum on the right-hand side is over the components present in
$z$.  We write $u \prec v$ to denote $u \le {\rm const}\, v$ with a
constant whose value is unimportant.

Consider first the case $\sum_k [\varphi_{i_k}]+|\beta|_1 > d_{+}$,
for which $\nabla^\beta r_z=\nabla^\beta h_z$.  By definition of
$d_+'$ in \refeq{dplusprimedef}, $\sum_k [\varphi_{i_k}]+|\beta|_1
\geq d_{+}'$.  We claim that the contribution to the right-hand side
of \eqref{e:rhognew} due to this case is
\begin{align}
\label{e:phijcase1}
    & \prec
    L^{-d_{+}'}
    \|h\|_{\Phi'( X_+)},
\end{align}
as required.  In fact, here there is no dependence on $R^{-1}{\rm
diam}(X)$ in the constant, and the hypothesis on $p_\Phi$ ensures that
there are sufficiently many derivatives in the norm of $h$.  The
potentially dangerous factor $(cK)^z$ is uniformly bounded when $p(z)$
is uniformly bounded, in particular with $p(z) \le d_+'/[\varphi_{\rm
min}]$.  On the other hand, when $p(z) > d_+'/[\varphi_{\rm min}]$,
the excess $(cK)^{p(z)-d_+'/[\varphi_{\rm min}]}$ is more than
compensated by the number of excess powers of $L^{-1}$ from
\refeq{rhognew}, namely $\sum_k [\varphi_{i_k}]+|\beta|_1 - d_+' \ge
p(z)[\varphi_{\rm min}] -d_+'$, for large $L$.

For the case $\sum_k [\varphi_{i_k}]+|\beta|_1 \le d_{+}$, we write
$t=|\beta|_1$ and $s=d_{+}-\sum_k [\varphi_{i_k}] \ge t$.  In this
case, $p(z)$ must be uniformly bounded, and hence so is the factor
$(cK)^z$ in \refeq{rhognew}.  By Lemma~\ref{lem:taylor-theorem}, there
exists $\bar c$, depending on $R^{-1}{\rm diam}(X)$, such that
\begin{align}
\label{e:tayerrest}
      |\nabla^\beta   r_{z} |
    &\le
    \bar {c}  \sup_{|\alpha |=s +1}
    R^{s-t+1} \sup_z |\nabla^{\alpha} h_{z} |
      \leq
     \bar {c} R^{s-t+1} (R')^{-s-1} (\h')^{z}
     \|h \|_{\Phi'(X_+)},
\end{align}
(the power of $R$ in the first line arises from the binomial
coefficient in \eqref{e:Tayrem2}, and it is here that the constant
develops its dependence on $R^{-1}{\rm diam} (X)$) and hence
\begin{align}
     (\h')^{-z}|\nabla_{R'}^\beta r_{z} |
    &\le
    \bar {c} R^{s-t+1} (R')^{t-s-1}
     \|h \|_{\Phi'(X_+)}
     \prec
    \bar {c} L^{t-s-1}
     \|h \|_{\Phi'(X_+)}.
\end{align}
Thus the contribution to \eqref{e:rhognew} due to this case is
\begin{equation}
    \prec
    \bar{c}
    L^{-\sum_k [\varphi_{i_k}]-t+t-s-1}
     \|h \|_{\Phi'(X_+)}
     =
    \bar{c} L^{-d_{+}-1}
     \|h \|_{\Phi'( X_+)}
.
\end{equation}
Since $d_+ +1 \ge d_+'$ by the definition of $d_{+'}$, this completes
the proof for the case of $\Zd$.

The torus case follows from the $\Zd$ case by the coordinate patch
assumption, once we choose $L$ large enough to ensure that the set
$\cup_{z \in {\mathbf X}_+}S_s(a,z)$ lies in a coordinate patch if
$X_+$ does.  This is possible because $j<n$ and hence there is a gap
of diameter at least $L$ preventing $X_+$ from ``wrapping around'' the
torus, whereas the enlargement of $X_+$ due to the set $S_s(a,z)$
depends only on $d_+$.  This enlargement cannot wrap around the torus
if $L$ is large enough.
\end{proof}

\section*{Acknowledgements}

The work of both authors was supported in part by NSERC of Canada.
DB gratefully acknowledges the support and hospitality of the
Institute for Advanced Study at Princeton and of Eurandom during part
of this work.
GS gratefully acknowledges the support and hospitality of the Institut
Henri Poincar\'e, where part of this work was done.
We thank Beno\^it Laslier for insightful comments which led to important
corrections, and also to simplifications in the proofs of
Propositions~\ref{prop:LTsymexists} and \ref{prop:Locbd}.
We also thank an anonymous referee for numerous pertinent suggestions.

\bibliography{../../bibdef/bib}
\bibliographystyle{plain}

\end{document}